\documentclass[10pt,oneside]{amsart}

\usepackage{amsmath, amssymb, graphics, setspace}

\newcommand{\mathsym}[1]{{}}
\newcommand{\unicode}[1]{{}}

\usepackage{amssymb}
\usepackage{amsfonts}
\usepackage{amscd}
\usepackage[matrix, arrow, curve]{xy}
\usepackage{graphicx}

\numberwithin{equation}{section}
\newtheorem{theorem}{Theorem}[section]

\newtheorem{proposition}[theorem]{Proposition}

\newtheorem{lemma}[theorem]{Lemma}
\theoremstyle{definition}

\theoremstyle{remark}
\newtheorem*{remark}{Remark}

\begin{document}
\newcommand{\M}{\mathcal{M}}
\newcommand{\F}{\mathcal{F}}

\newcommand{\Teich}{\mathcal{T}_{g,N+1}^{(1)}}
\newcommand{\T}{\mathrm{T}}
\newcommand{\corr}{\bf}
\newcommand{\vac}{|0\rangle}
\newcommand{\Ga}{\Gamma}
\newcommand{\new}{\bf}
\newcommand{\define}{\def}
\newcommand{\redefine}{\def}
\newcommand{\Cal}[1]{\mathcal{#1}}
\renewcommand{\frak}[1]{\mathfrak{{#1}}}
\newcommand{\Hom}{\rm{Hom}\,}
\newcommand{\refE}[1]{(\ref{E:#1})}
\newcommand{\refCh}[1]{Chapter~\ref{Ch:#1}}
\newcommand{\refS}[1]{Section~\ref{S:#1}}
\newcommand{\refSS}[1]{Section~\ref{SS:#1}}
\newcommand{\refT}[1]{Theorem~\ref{T:#1}}
\newcommand{\refO}[1]{Observation~\ref{O:#1}}
\newcommand{\refP}[1]{Proposition~\ref{P:#1}}
\newcommand{\refD}[1]{Definition~\ref{D:#1}}
\newcommand{\refC}[1]{Corollary~\ref{C:#1}}
\newcommand{\refL}[1]{Lemma~\ref{L:#1}}
\newcommand{\R}{\ensuremath{\mathbb{R}}}
\newcommand{\C}{\ensuremath{\mathbb{C}}}
\newcommand{\N}{\ensuremath{\mathbb{N}}}
\newcommand{\Q}{\ensuremath{\mathbb{Q}}}
\renewcommand{\P}{\ensuremath{\mathcal{P}}}
\newcommand{\Z}{\ensuremath{\mathbb{Z}}}
\newcommand{\kv}{{k^{\vee}}}
\renewcommand{\l}{\lambda}
\newcommand{\gb}{\overline{\mathfrak{g}}}
\newcommand{\dt}{\tilde d}     
\newcommand{\hb}{\overline{\mathfrak{h}}}
\newcommand{\g}{\mathfrak{g}}
\newcommand{\h}{\mathfrak{h}}
\newcommand{\gh}{\widehat{\mathfrak{g}}}
\newcommand{\ghN}{\widehat{\mathfrak{g}_{(N)}}}
\newcommand{\gbN}{\overline{\mathfrak{g}_{(N)}}}
\newcommand{\tr}{\mathrm{tr}}
\newcommand{\gln}{\mathfrak{gl}(n)}
\newcommand{\son}{\mathfrak{so}(n)}
\newcommand{\spnn}{\mathfrak{sp}(2n)}
\newcommand{\sln}{\mathfrak{sl}}
\newcommand{\sn}{\mathfrak{s}}
\newcommand{\so}{\mathfrak{so}}
\newcommand{\spn}{\mathfrak{sp}}
\newcommand{\tsp}{\mathfrak{tsp}(2n)}
\newcommand{\gl}{\mathfrak{gl}}
\newcommand{\slnb}{{\overline{\mathfrak{sl}}}}
\newcommand{\snb}{{\overline{\mathfrak{s}}}}
\newcommand{\sob}{{\overline{\mathfrak{so}}}}
\newcommand{\spnb}{{\overline{\mathfrak{sp}}}}
\newcommand{\glb}{{\overline{\mathfrak{gl}}}}
\newcommand{\Hwft}{\mathcal{H}_{F,\tau}}
\newcommand{\Hwftm}{\mathcal{H}_{F,\tau}^{(m)}}
\newcommand{\ad}{{\rm{ad}\,}}

\newcommand{\car}{{\mathfrak{h}}}    
\newcommand{\bor}{{\mathfrak{b}}}    
\newcommand{\nil}{{\mathfrak{n}}}    
\newcommand{\vp}{{\varphi}}
\newcommand{\bh}{\widehat{\mathfrak{b}}}  
\newcommand{\bb}{\overline{\mathfrak{b}}}  
\newcommand{\Vh}{\widehat{\mathcal V}}
\newcommand{\KZ}{Kniz\-hnik-Zamo\-lod\-chi\-kov}
\newcommand{\TUY}{Tsuchia, Ueno  and Yamada}
\newcommand{\KN} {Kri\-che\-ver-Novi\-kov}
\newcommand{\pN}{\ensuremath{(P_1,P_2,\ldots,P_N)}}
\newcommand{\xN}{\ensuremath{(\xi_1,\xi_2,\ldots,\xi_N)}}
\newcommand{\lN}{\ensuremath{(\lambda_1,\lambda_2,\ldots,\lambda_N)}}
\newcommand{\iN}{\ensuremath{1,\ldots, N}}
\newcommand{\iNf}{\ensuremath{1,\ldots, N,\infty}}

\newcommand{\tb}{\tilde \beta}
\newcommand{\tk}{\tilde \varkappa}
\newcommand{\ka}{\kappa}
\renewcommand{\k}{\varkappa}
\newcommand{\ce}{{c}}

\newcommand{\Pif} {P_{\infty}}
\newcommand{\Pinf} {P_{\infty}}
\newcommand{\PN}{\ensuremath{\{P_1,P_2,\ldots,P_N\}}}
\newcommand{\PNi}{\ensuremath{\{P_1,P_2,\ldots,P_N,P_\infty\}}}
\newcommand{\Fln}[1][n]{F_{#1}^\lambda}
\newcommand{\tang}{\mathrm{T}}
\newcommand{\Kl}[1][\lambda]{\can^{#1}}
\newcommand{\A}{\mathcal{A}}
\newcommand{\U}{\mathcal{U}}
\newcommand{\V}{\mathcal{V}}
\newcommand{\W}{\mathcal{W}}
\renewcommand{\O}{\mathcal{O}}
\newcommand{\Ae}{\widehat{\mathcal{A}}}
\newcommand{\Ah}{\widehat{\mathcal{A}}}
\newcommand{\La}{\mathcal{L}}
\newcommand{\Le}{\widehat{\mathcal{L}}}
\newcommand{\Lh}{\widehat{\mathcal{L}}}
\newcommand{\eh}{\widehat{e}}
\newcommand{\Da}{\mathcal{D}}
\newcommand{\kndual}[2]{\langle #1,#2\rangle}
\newcommand{\cins}{\frac 1{2\pi\mathrm{i}}\int_{C_S}}
\newcommand{\cinsl}{\frac 1{24\pi\mathrm{i}}\int_{C_S}}
\newcommand{\cinc}[1]{\frac 1{2\pi\mathrm{i}}\int_{#1}}
\newcommand{\cintl}[1]{\frac 1{24\pi\mathrm{i}}\int_{#1 }}
\newcommand{\w}{\omega}
\newcommand{\ord}{\operatorname{ord}}
\newcommand{\res}{\operatorname{res}}
\newcommand{\nord}[1]{:\mkern-5mu{#1}\mkern-5mu:}
\newcommand{\codim}{\operatorname{codim}}

\newcommand{\Fn}[1][\lambda]{\mathcal{F}^{#1}}
\newcommand{\Fl}[1][\lambda]{\mathcal{F}^{#1}}
\renewcommand{\Re}{\mathrm{Re}}

\newcommand{\ha}{H^\alpha}

\define\ldot{\hskip 1pt.\hskip 1pt}
\define\ifft{\qquad\text{if and only if}\qquad}
\define\a{\alpha}
\redefine\d{\delta}
\define\w{\omega}
\define\ep{\epsilon}
\redefine\b{\beta} \redefine\t{\tau} \redefine\i{{\,\mathrm{i}}\,}
\define\ga{\gamma}
\define\cint #1{\frac 1{2\pi\i}\int_{C_{#1}}}
\define\cintta{\frac 1{2\pi\i}\int_{C_{\tau}}}
\define\cintt{\frac 1{2\pi\i}\oint_{C}}
\define\cinttp{\frac 1{2\pi\i}\int_{C_{\tau'}}}
\define\cinto{\frac 1{2\pi\i}\int_{C_{0}}}
\define\cinttt{\frac 1{24\pi\i}\int_C}
\define\cintd{\frac 1{(2\pi \i)^2}\iint\limits_{C_{\tau}\,C_{\tau'}}}
\define\dintd{\frac 1{(2\pi \i)^2}\iint\limits_{C\,C'}}
\define\cintdr{\frac 1{(2\pi \i)^3}\int_{C_{\tau}}\int_{C_{\tau'}}
\int_{C_{\tau''}}}
\define\im{\operatorname{Im}}
\define\re{\operatorname{Re}}
\define\res{\operatorname{res}}
\redefine\deg{\operatornamewithlimits{deg}}
\define\ord{\operatorname{ord}}
\define\rank{\operatorname{rank}}
\define\fpz{\frac {d }{dz}}
\define\dzl{\,{dz}^\l}
\define\pfz#1{\frac {d#1}{dz}}

\define\K{\Cal K}
\define\U{\Cal U}
\redefine\O{\Cal O}
\define\He{\text{\rm H}^1}
\redefine\H{{\mathrm{H}}}
\define\Ho{\text{\rm H}^0}
\define\A{\Cal A}
\define\Do{\Cal D^{1}}
\define\Dh{\widehat{\mathcal{D}}^{1}}
\redefine\L{\Cal L}
\newcommand{\ND}{\ensuremath{\mathcal{N}^D}}
\redefine\D{\Cal D^{1}}
\define\KN {Kri\-che\-ver-Novi\-kov}
\define\Pif {{P_{\infty}}}
\define\Uif {{U_{\infty}}}
\define\Uifs {{U_{\infty}^*}}
\define\KM {Kac-Moody}
\define\Fln{\Cal F^\lambda_n}
\define\gb{\overline{\mathfrak{ g}}}
\define\G{\overline{\mathfrak{ g}}}
\define\Gb{\overline{\mathfrak{ g}}}
\redefine\g{\mathfrak{ g}}
\define\Gh{\widehat{\mathfrak{ g}}}
\define\gh{\widehat{\mathfrak{ g}}}
\define\Ah{\widehat{\Cal A}}
\define\Lh{\widehat{\Cal L}}
\define\Ugh{\Cal U(\Gh)}
\define\Xh{\hat X}
\define\Tld{...}
\define\iN{i=1,\ldots,N}
\define\iNi{i=1,\ldots,N,\infty}
\define\pN{p=1,\ldots,N}
\define\pNi{p=1,\ldots,N,\infty}
\define\de{\delta}

\define\kndual#1#2{\langle #1,#2\rangle}
\define \nord #1{:\mkern-5mu{#1}\mkern-5mu:}
\newcommand{\MgN}{\mathcal{M}_{g,N}} 
\newcommand{\MgNeki}{\mathcal{M}_{g,N+1}^{(k,\infty)}} 
\newcommand{\MgNeei}{\mathcal{M}_{g,N+1}^{(1,\infty)}} 
\newcommand{\MgNekp}{\mathcal{M}_{g,N+1}^{(k,p)}} 
\newcommand{\MgNkp}{\mathcal{M}_{g,N}^{(k,p)}} 
\newcommand{\MgNk}{\mathcal{M}_{g,N}^{(k)}} 
\newcommand{\MgNekpp}{\mathcal{M}_{g,N+1}^{(k,p')}} 
\newcommand{\MgNekkpp}{\mathcal{M}_{g,N+1}^{(k',p')}} 
\newcommand{\MgNezp}{\mathcal{M}_{g,N+1}^{(0,p)}} 
\newcommand{\MgNeep}{\mathcal{M}_{g,N+1}^{(1,p)}} 
\newcommand{\MgNeee}{\mathcal{M}_{g,N+1}^{(1,1)}} 
\newcommand{\MgNeez}{\mathcal{M}_{g,N+1}^{(1,0)}} 
\newcommand{\MgNezz}{\mathcal{M}_{g,N+1}^{(0,0)}} 
\newcommand{\MgNi}{\mathcal{M}_{g,N}^{\infty}} 
\newcommand{\MgNe}{\mathcal{M}_{g,N+1}} 
\newcommand{\MgNep}{\mathcal{M}_{g,N+1}^{(1)}} 
\newcommand{\MgNp}{\mathcal{M}_{g,N}^{(1)}} 
\newcommand{\Mgep}{\mathcal{M}_{g,1}^{(p)}} 
\newcommand{\MegN}{\mathcal{M}_{g,N+1}^{(1)}} 

\define \sinf{{\widehat{\sigma}}_\infty}
\define\Wt{\widetilde{W}}
\define\St{\widetilde{S}}
\newcommand{\SigmaT}{\widetilde{\Sigma}}
\newcommand{\hT}{\widetilde{\frak h}}
\define\Wn{W^{(1)}}
\define\Wtn{\widetilde{W}^{(1)}}
\define\btn{\tilde b^{(1)}}
\define\bt{\tilde b}
\define\bn{b^{(1)}}
\define \ainf{{\frak a}_\infty} 

%
\define\eps{\varepsilon}    
\newcommand{\e}{\varepsilon}
\define\doint{({\frac 1{2\pi\i}})^2\oint\limits _{C_0}
       \oint\limits _{C_0}}                            
\define\noint{ {\frac 1{2\pi\i}} \oint}   
\define \fh{{\frak h}}     
\define \fg{{\frak g}}     
\define \GKN{{\Cal G}}   
\define \gaff{{\hat\frak g}}   
\define\V{\Cal V}
\define \ms{{\Cal M}_{g,N}} 
\define \mse{{\Cal M}_{g,N+1}} 
\define \tOmega{\Tilde\Omega}
\define \tw{\Tilde\omega}
\define \hw{\hat\omega}
\define \s{\sigma}
\define \car{{\frak h}}    
\define \bor{{\frak b}}    
\define \nil{{\frak n}}    
\define \vp{{\varphi}}
\define\bh{\widehat{\frak b}}  
\define\bb{\overline{\frak b}}  
\define\KZ{Knizhnik-Zamolodchikov}
\define\ai{{\alpha(i)}}
\define\ak{{\alpha(k)}}
\define\aj{{\alpha(j)}}
\newcommand{\calF}{{\mathcal F}}
\newcommand{\ferm}{{\mathcal F}^{\infty /2}}
\newcommand{\laxgl}{\overline{\mathfrak{gl}}}
\newcommand{\laxsl}{\overline{\mathfrak{sl}}}
\newcommand{\laxso}{\overline{\mathfrak{so}}}
\newcommand{\laxsp}{\overline{\mathfrak{sp}}}
\newcommand{\laxs}{\overline{\mathfrak{s}}}
\newcommand{\laxg}{\overline{\frak g}}
\newcommand{\bgl}{\laxgl(n)}
\newcommand{\tX}{\widetilde{X}}
\newcommand{\tY}{\widetilde{Y}}
\newcommand{\tZ}{\widetilde{Z}}



\title[Some reductions]{Some reductions of rank 2 and genera 2 and 3 Hitchin systems}
\author[O.K. Sheinman]{Oleg K. Sheinman}
\thanks{}
\address{Department of Geometry and Topology, Steklov Mathematical Institute, Moscow, and Independent University of Moscow}
\maketitle
\begin{abstract}
Certain reductions of the rank~2, genera~2 and~3 Hitchin systems are considered, which are shown to give an integrable system of 2, resp. 3, interacting points on the line. It is shown that the reduced systems are particular cases of a certain universal integrable system related to the Lagrange interpolation polynomial. Admissibility of the reduction is proved using computer technique. The corresponding codes are given in the text.
\end{abstract}
\tableofcontents
\section{Introduction}
General procedures of finding the algebraic-geometric solutions, and the action-angle coordinates for Hitchin systems are proposed in \cite{Klax}. For the rank~2 genus~2 systems the problem had been considered earlier in \cite{Gaw} based on the analogy with the Neumann system, and the results of \cite{Prev}. However, if we ask what a dynamical system corresponds to a given Hitchin Hamiltonian\footnote{The author is grateful to I.M.Krichever for this question}, say, in the Tyurin parametrization, we will find out that it is a problem even write down the corresponding equations because the expressions for the Hamiltonians consist of thousands of symbols.

In the present work we consider \emph{reductions} of Hitchin systems, i.e. their restrictions to invariant subvarieties of a positive codimension. On this way we obtain some manageable  rank~2, genera~2 and~3 systems, certain relations for their dynamical variables, and some particular solutions to the original systems.

It turns out to be that on the genus $g$ curve the reduced system is integrable and coincides with the universal integrable system related to the Lagrange interpolation polynomial of degree $g-1$.

The paper is organized as follows. In \refS{Red} we describe the reduction of the Hitchin systems of rank~2, genera~2 and~3, prove it to be an admissible reduction, observe that the reduced systems are completely integrable and  find some their particular solutions (which are solutions to the original Hitchin system as well). The proof of admissibility is computational. The codes of the corresponding programs are given in Sections \ref{S:Genus2}, \ref{S:Genus3}.

\refS{Lagrange} is devoted to the integrable system related to the Lagrange interpolation polynomial.

The author is grateful to S.P.Novikov who multiply put the question of detailed investigation of Hitchin systems of small ranks and genera, and to I.M.Krichever for discussions of setting the problem and some results.

\section{Reduction, genera 2 and 3}\label{S:Red}
Let $\Sigma$ be a genus $g$ hyperelliptic curve given by the equation $y^2=P_{2g+1}(x)$ where $P_{2g+1}$ is a degree $2g+1$ polynomial. Consider a Lax operator with the spectral parameter on this curve, taking values in the full linear algebra~$\gl(2)$, given in the Tyurin parametrization \cite{Klax} (see  \cite{Sh_DAN_LOA&gr}--\cite{Tyur65} for the further developments):
\[
  L(x,y)=\sum_{i=0}^{g-1}L_ix^i+\sum_{s=1}^{2g}\a_s\b_s^T\frac{y+b_s}{x-a_s}\ ,
\]
where $(a_s,b_s)\in\Sigma$, $\a_s,\b_s\in\C^2$, $\b_s^T\a_s=0$, $L(a_s,b_s)\a_s=\kappa_s\a_s$, $\kappa_s\in\C$  for any $s=1,\ldots,2g$ (the upper $T$ denotes transposing here). Observe that there arise no singularity in the above eigenvalue conditions due to the assumption $\b_s^T\a_s=0$, and also because the points $a_s$ ($s=1,\ldots,2g$) are assumed to be different. By fixing a gauge, we set the matrix formed by the vectors $\a_{2g-1},\a_{2g}$ to the unit $2\times 2$ matrix. Below, $g=2,3$.

A rank~2 Hitchin system on $\Sigma$ is given by the canonical symplectic structure $\w=\sum_{s=1}^{2g}\a_s^T\wedge\b_s+\sum_{s=1}^{2g} a_s\wedge\kappa_s$, and the Hamiltonians of the form
\[
     H_{k,m}~=~\res_{z=0} z^{m-1}\tr L(z)^k {\mathrm d}x(z)/{y(z)} ,
\]
where $z=\frac{1}{\sqrt{x}}$ is a local parameter in the neighborhood of the point $x=\infty$, and $k,m\in\Z$.

In the case $g=2$ we consider the canonical equations corresponding to the Hamiltonian~$H_{2,2}$. Assuming $p_1=0$ where $y^2=x^5+p_1x^4+\ldots+p_5$ is the equation of the curve we state
\begin{proposition}\label{P:gen2}
For a genus~2 curve $\Sigma$ the reduction  $\a_{11}=\a_{22}=0$, $\b_1=\b_2=\b_3=\b_4=0$ is admissible (compatible with the system), and the corresponding reduced system has the form
\[
\begin{aligned}
&\dot a_1=-\frac{2 \left(a_1 \kappa _4+a_4 \left(-2 \kappa _1+\kappa _4\right)\right)}{\left(a_1-a_4\right){}^2},\quad \dot \kappa_1=\frac{ \kappa _1-\kappa _4}{a_1-a_4}\,\dot a_1    \\
&\dot a_4=-\frac{2 \left(a_4 \kappa _1+a_1 \left(\kappa _1-2 \kappa _4\right)\right) }{\left(a_1-a_4\right){}^2},\quad\ \ \dot \kappa_4=\frac{ \kappa _1-\kappa _4}{a_1-a_4}\,\dot a_4     \\
&\dot a_2=-\frac{2 \left(a_2 \kappa _3+a_3 \left(-2 \kappa _2+\kappa _3\right)\right)}{\left(a_2-a_3\right){}^2},\quad \dot \kappa_2=\frac{ \kappa _2-\kappa _3}{a_2-a_3}\,\dot a_2       \\
&\dot a_3=-\frac{2 \left(a_3 \kappa _2+a_2 \left(\kappa _2-2 \kappa _3\right)\right) }{\left(a_2-a_3\right){}^2},\quad\ \ \dot\kappa_3=\frac{ \kappa_2-\kappa_3}{a_2-a_3}\,\dot a_3     \\
&\phantom{aaaa}\dot \alpha_{12}=\alpha _{12}\frac{b_2+b_3}{a_2-a_3}\,\dot a_3,\quad\dot\alpha_{21}=\alpha _{21}\frac{b_1+b_4}{a_1-a_4}\,\dot a_4 .
\end{aligned}
\]
The system of the first $8$ equations (for the system of points) is Hamiltonian with the Hamiltonian
\begin{equation}\label{E:H_gen2}
\begin{aligned}
H_{2,2}^{(r)}&= 2\left(\kappa _1-\kappa _4\right)\left(a_4\kappa _1-a_1\kappa _4\right)/\left(a_1-a_4\right){}^2
\\
&+ 2\left(\kappa _2-\kappa _3\right)\left(a_3\kappa _2-a_2\kappa _3\right)/\left(a_2-a_3\right){}^2
\end{aligned}
\end{equation}
\end{proposition}
The proof consists of a direct calculation using the program "Wolfram Mathematica":  it turns out to be that the derivatives of the corresponding variables by virtue of the system vanish after plugging  $\a_{11}\!=\!\a_{22}\!=0$, $\b_1=\b_2=\b_3=\b_4=0$. The programs for this and other calculations are given below (in \refS{Genus2} for genus~2, and in \refS{Genus3} for genus~3).

The obtained system of equations splits into the two consistent subsystems: one of them for unknowns $a_1$, $a_4$, $\kappa_1$, $\kappa_4$, $\a_{21}$, and another for unknowns $a_2$, $a_3$, $\kappa_2$, $\kappa_3$, $\a_{12}$. Both can be resolved similarly. For example, we have for the first one
$
   \frac{\dot\kappa_1-\dot\kappa_4}{\kappa_1-\kappa_4} = \frac{\dot a_1-\dot a_4}{a_1-a_4}\ ,
$
which implies $\kappa_1=c_1a_1+c_2$, $\kappa_4=c_1a_4+c_3$ ($c_1$, $c_2$, $c_3$ are constants). For $c_1=c_2=c_3$ the arising equations for $a_1$, $a_4$ can be completely integrated: $a_1+a_4+2a_4^2=ct+\tilde{\tilde c}$, $2a_4+\ln(a_1+a_4)+{\tilde c}=0$, where $c$, ${\tilde c}$, $\tilde{\tilde c}$ are constants. Next, the equation for $\a_{21}$ can be transformed to the form $d\ln\a_{21}=\frac{b_1+b_4}{a_1-a_4}\, da_4$. The right hand side can be considered as a function of only $a_4$, which is, moreover, already known, and we have $\a_{21}=\exp\int\frac{b_1+b_4}{a_1-a_4}\, da_4$.


Besides the reduction  $\a_{11}\!=\!\a_{22}\!=\! 0$ there is another one  $\a_{12}\!=\!\a_{21}\!=\! 0$ which descends to the previous by means of permutation of the points $a_1$ and $a_2$.

For a genus~3 curve
\[
y^2=x^7+ p_1x^6+p_2x^5+p_3x^4+p_4x^3+p_5x^2+p_6x+p_7
\]
we fix a gauge by requiring that the vectors $\a_5$ and $\a_6$ form the unit matrix. Then for $p_1=p_2=p_3=0$ the reduced system splits to the two (independent) triples of interacting points, and the equations for the parameters  $\a_{ij}$. For example for the Hamiltonian $H_{2,4}$ under the reduction
\[
\begin{aligned}
  &\alpha _{11}= 0,\alpha _{22}= 0,\alpha _{13}= 0,\alpha _{24}= 0,\\
  &\phantom{aaaaa}\beta _s= 0\quad (s=1,\ldots,6)
  \end{aligned}
\]
we obtain for the points $a_1$, $a_3$, $a_6$ (for which the first coordinates of the corresponding vectors $\a_s$ are equal to zero) a Hamiltonian system with the Hamiltonian
\[
\begin{aligned}
&H_{2,4}^{(r)}=\frac{\left(a_6^2 \left(-\kappa _1+\kappa _3\right)+a_3^2 \left(\kappa _1-\kappa _6\right)+a_1^2 \left(-\kappa _3+\kappa _6\right)\right){}^2}{\left(a_1-a_3\right){}^2 \left(a_1-a_6\right){}^2 \left(a_3-a_6\right){}^2} +
\\
&+
\frac{2 \left(a_6 \left(\kappa _3-\kappa _1\right)+a_3 \left(\kappa _1-\kappa _6\right)+a_1 \left(\kappa _6-\kappa _3\right)\right)}{\left(a_1-a_3\right) \left(a_1-a_6\right) \left(a_3-a_6\right)}\times
\\
&\times\frac{a_6^2\left(a_1\kappa_3-a_3\kappa_1\right)+a_3^2 \left(a_6 \kappa _1-a_1 \kappa _6\right)+a_1^2\left(a_3\kappa_6-a_6\kappa_3\right)}{\left(a_1-a_3\right) \left(a_1-a_6\right) \left(a_3-a_6\right)}
\end{aligned}
\]
The following relations between $\dot\kappa _i$ and ${\dot a}_i$ hold:
\begin{equation}\label{E:dynamics}
\begin{aligned}
\dot\kappa_1&=\frac{\left(a_1-a_6\right){}^2\left(\kappa_1-\kappa _3\right)-\left(a_1-a_3\right){}^2\left(\kappa _1-\kappa _6\right)}{\left(a_1-a_3\right)\left(a_1-a_6\right)\left(a_3-a_6\right)}\dot a_1
\\
\dot\kappa_3&=\frac{\left(a_3-a_6\right){}^2\left(\kappa _1-\kappa _3\right)+\left(a_1-a_3\right){}^2\left(\kappa _3-\kappa _6\right)}{\left(a_1-a_3\right)\left(a_1-a_6\right)\left(a_3-a_6\right)}\dot a_3
\\
\dot\kappa_6&=\frac{\left(a_6-a_1\right)^2\left(\kappa_3-\kappa _6\right)-\left(a_6-a_3\right){}^2\left(\kappa _1-\kappa _6\right)}{\left(a_1-a_3\right)\left(a_1-a_6\right)\left(a_3-a_6\right)}\dot a_6
\end{aligned}
\end{equation}
Also we have:
\begin{equation}\label{E:dynamics1}
  \dot a_1+\dot a_3+\dot a_6=
  \frac{2 \left(a_6 \left(-\kappa _1+\kappa _3\right)+a_3 \left(\kappa _1-\kappa _6\right)+a_1 \left(-\kappa _3+\kappa _6\right)\right)}{\left(a_1-a_3\right) \left(a_1-a_6\right) \left(a_3-a_6\right)}.
\end{equation}
The relations \refE{dynamics} are proven below (see \refL{commfact1} and after that). As for \refE{dynamics1}, it still is only a result of computation.

Under assumption $\kappa_1=\kappa_3=\kappa_6=K(=const)$ we obtain the following equations:
\[
\begin{aligned}
\dot a_1&=\frac{2K}{\left(a_1-a_3\right) \left(a_1-a_6\right)}
\\
\dot a_3&=-\frac{2K}{\left(a_1-a_3\right) \left(a_3-a_6\right)}
\\
\dot a_6&=\frac{2K}{ \left(a_1-a_6\right)\left(a_3-a_6\right)}
\end{aligned}
\]
and
\[
\begin{aligned}
\frac{\dot\a_{21}}{\a_{21}}&=\frac{2K\left(b_1+b_6\right) }{ \left(a_1-a_6\right){}^2\left(a_3-a_6\right)}-\frac{2K\left(b_1+b_3\right) }{\left(a_1-a_3\right){}^2\left(a_3-a_6\right)}
\\
\frac{\dot\a_{23}}{\a_{23}}&=\frac{2K\left(b_3+b_6\right) }{\left(a_1-a_6\right)\left(a_3-a_6\right){}^2}-\frac{2K\left(b_1+b_3\right) }{\left(a_1-a_3\right){}^2\left(a_1-a_6\right) },
\end{aligned}
\]
in particular it follows that
$
\dot a_1+ \dot a_3+ \dot a_6=0
$.
The variables in the system of equations on $a_s$ do separate: setting
\[
a_1+ a_3+ a_6=C_1
\]
we obtain
\[
(\dot a_1-\dot a_3)(a_1-a_3)=-3(\dot a_1+\dot a_3)(a_1+a_3-2C_1/3)
\]
which is resolved in the form:
\[
\begin{aligned}
 & a_1^3+f_1C_1a_1^2+f_2C_1^2a_1=-2Kt+C_2, \\
 &a_3=e^{\pm 2\pi i/3}a_1+f_0C_1,\quad a_6=C_1-a_1-a_3,
\end{aligned}
\]
where $f_0,f_1,f_2\in\C$ are known coefficients (the explicit expressions for them are omitted here), $C_1,C_2$ are the integration constants. In particular, for $C_1=0$ we obtain
\[
a_1=\sqrt[3]{-2Kt+C_2},\quad a_3=e^{\pm 2\pi i/3}a_1,\quad a_1+a_3+a_6=0.
\]
The  variables $\a_{is}$ can be expressed via quadratures, similar to the genus~2 case.
\begin{proposition}
The reduced systems are completely integrable. The corresponding full system of integrals in involution for $g=2$ (and the system of the points $\{a_1,a_4\}$) is as follows:
\[
F_0= \frac{a_1 \kappa _4-a_4 \kappa _1}{a_1-a_4},
\quad
F_1=\frac{\kappa _1-\kappa _4}{a_1-a_4},
\]
and for $g=3$ as follows:
\[
\begin{aligned}
&F_0=\frac{a_1 a_6 \left(a_6-a_1\right) \kappa _3+a_3^2 \left(a_6 \kappa _1-a_1 \kappa _6\right)+a_3 \left(a_1^2 \kappa _6-a_6^2 \kappa _1\right)}{\left(a_1-a_3\right)\left(a_1-a_6\right)\left(a_3-a_6\right)}
\\
&F_1=\frac{a_6^2 \left(-\kappa _1+\kappa _3\right)+a_3^2 \left(\kappa _1-\kappa _6\right)+a_1^2 \left(-\kappa _3+\kappa _6\right)}{\left(a_1-a_3\right)\left(a_1-a_6\right) \left(a_3-a_6\right)}
\\
&F_2=\frac{a_6 \left(-\kappa _1+\kappa _3\right)+a_3 \left(\kappa _1-\kappa _6\right)+a_1 \left(-\kappa _3+\kappa _6\right)}{\left(a_1-a_3\right)\left(a_1-a_6\right) \left(a_3-a_6\right)}
\end{aligned}
\]
\end{proposition}
\begin{proof}
It is proven in the next section that the integrals $F_0$, $F_1$ for $g=2$, and the integrals $F_0$, $F_1$, $F_2$ for $g=3$, are in involution with respect to the canonical Poisson bracket (and obviously independent). The proposition follows now from the fact that $H_{2,2}^{(r)}=2F_0F_1$ for $g=2$ (by $H_{2,2}^{(r)}$ we mean here the part of the Hamiltonian \refE{H_gen2} related to the points $a_1$, $a_4$), and $H_{2,4}^{(r)}=F_1^2+2F_0F_2$ for $g=3$. Hence $H_{2,2}^{(r)}$ (resp., $H_{2,4}^{(r)}$) is in involution with the basis integrals.
\end{proof}
\begin{remark}
It looks like the above vanishing assumptions on the coefficients $p_i$ of the equation of the curve do not affect the reduced system. At least, for $g=2$ we have $H_{2,2}^{(r)}=2(F_0+p_1F_1)F_1$ for an arbitrary $p_1$, i.e. the basis integrals are the same.
\end{remark}

The above listed integrals are related to the universal integrable system introduced in \refS{Lagrange} (\refP{main}). We conclude the section with the following conjecture. \emph{ The space where a half of the vectors $\a_s$ is proportional to $(1,0)^T$, the remainder to the $(0,1)^T$, and  $\b_s=0$ for every $s$, is an invariant subspace of the system. The reduced system splits into two independent  Hamiltonian completely integrable systems. Each of two reduced systems is given by plugging the corresponding points to the integrals of the universal integrable system of \refS{Lagrange}.} The above obtained results prove the conjecture for $g=2,3$.

\section{Integrable system related to the Lagrange interpolation polynomial}
\label{S:Lagrange}
Let $F(x)=\sum_{i=0}^{n-1} F_ix^i$ be a polynomial taking values $\kappa_1,\ldots,\kappa_n$ at the different points $a_1,\ldots,a_n$ (called \emph{nodes} in the context of interpolation). Consider its coefficients as functions of the sets of nodes and values:
\[
    F_i=F_i(a_1,\ldots,a_n,\kappa_1,\ldots,\kappa_n),\quad i=0,\ldots,n-1.
\]
By \emph{canonical Poisson bracket} we mean here the Poisson bracket in $\C^{2n}$ equipped with coordinates $a_1,\ldots,a_n$, $\kappa_1,\ldots,\kappa_n$,  defined by means the following relations:
\[
   \{ a_r,\kappa_s\}=\delta_{rs},\quad \{ a_r,a_s\}=\{ \kappa_r,\kappa_s\}=0
\]
for all pairs $r,s$, where $\delta_{rs}$ is the Kronecker symbol.
The following statement is the main goal of this section.
\begin{proposition}\label{P:main}
The functions $F_i$ commute with respect to the canonical Poisson bracket.
\end{proposition}
The proof immediately follows from the following lemma:
\begin{lemma}\label{L:commfact}
There exist rational functions $M_k=M_k(a_1,\ldots,a_n,\kappa_1,\ldots,\kappa_n)$ such that for any $i=0,\ldots,n-1$
\[
  \frac{\partial F_i}{\partial a_k}=M_k\frac{\partial F_i}{\partial \kappa_k},\quad k=1,\ldots,n
\]
independently of $i$.
\end{lemma}
Indeed, by the lemma we obtain
\[
    \frac{\partial F_i}{\partial a_k}\frac{\partial F_j}{\partial \kappa_k}= M_k\frac{\partial F_i}{\partial \kappa_k}\frac{\partial F_j}{\partial \kappa_k},
\]
hence
\[
  \frac{\partial F_i}{\partial a_k}\frac{\partial F_j}{\partial \kappa_k}- \frac{\partial F_j}{\partial a_k}\frac{\partial F_i}{\partial \kappa_k}=0,
\]
for all $k=1,\ldots,n$, and
\[
  \{ F_i,F_j \}=\sum_{k=1}^n \left( \frac{\partial F_i}{\partial a_k}\frac{\partial F_j}{\partial \kappa_k}- \frac{\partial F_j}{\partial a_k}\frac{\partial F_i}{\partial \kappa_k}\right) =0.
\]
\begin{proof}[Proof of \refL{commfact}.]
It is obviously sufficient to prove the lemma for $k=1$.

Represent $F$ in the form of Lagrange interpolation polynomial \cite{Prasolov}:
\begin{equation}\label{E:Lagrange}
  F(x)=\sum_{s=1}^n \kappa_sf_s(x)
\end{equation}
where
\[
  f_s(x)=\frac{(x-a_1)\ldots(x-a_{s-1})(x-a_{s+1})\ldots(x-a_n)} {(a_s-a_1)\ldots(a_s-a_{s-1})(a_s-a_{s+1})\ldots(a_s-a_n)}.
\]
In particular,
\[
   f_1(x)=\frac{(x-a_2)\ldots(x-a_n)} {(a_1-a_2)\ldots(a_1-a_n)},
\]
hence
\[
    \frac{\partial}{\partial a_1}f_1(x)=(x-a_2)\ldots(x-a_n)\frac{\partial}{\partial a_1}\frac{1} {(a_1-a_2)\ldots(a_1-a_n)}.
\]
For $s>1$ we have
\[
    \frac{\partial}{\partial a_1} f_s(x)= \frac{(x-a_2)\ldots(x-a_{s-1})(x-a_{s+1})\ldots(x-a_n)} {(a_s-a_2)\ldots(a_s-a_{s-1})(a_s-a_{s+1})\ldots(a_s-a_n)} \cdot\frac{\partial}{\partial a_1}\frac{x-a_1} {a_s-a_1}.
\]
Since
\[
   \frac{\partial}{\partial a_1}\frac{x-a_1} {a_s-a_1}= \frac{x-a_s}{(a_s-a_1)^2},
\]
the gap in the nominator of the previous expression will get filled up, and we obtain
\[
   \frac{\partial}{\partial a_1} f_s(x)= \frac{(x-a_2)\ldots(x-a_n)}{(a_s-a_1)^2(a_s-a_2)\ldots(a_s-a_{s-1}) (a_s-a_{s+1})\ldots (a_s-a_n)}.
\]
We conclude that
\begin{equation}\label{E:comp1}
    \frac{\partial}{\partial a_1}F(x)=M_1'(x-a_2)\ldots(x-a_n)
\end{equation}
where
\[
\begin{aligned}
  M_1'&=\frac{\partial}{\partial a_1}\frac{\kappa_1} {(a_1-a_2)\ldots(a_1-a_n)} \\ &+ \sum_{s=2}^n \frac{\kappa_s}{(a_s-a_1)^2(a_s-a_2)\ldots(a_s-a_{s-1}) (a_s-a_{s+1})\ldots (a_s-a_n)}.
\end{aligned}
\]
Further on, by \refE{Lagrange} we obviously have
\begin{equation}\label{E:comp2}
   \frac{\partial}{\partial \kappa_1}F(x)= f_1(x) =\frac{(x-a_2)\ldots(x-a_n)} {(a_1-a_2)\ldots(a_1-a_n)}.
\end{equation}
Comparing \refE{comp1} and \refE{comp2}, and setting $M_1=M_1'(a_1-a_2)\ldots(a_1-a_n)$, we conclude that
\[
  \frac{\partial}{\partial a_1}F(x)=M_1 \frac{\partial}{\partial \kappa_1}F(x).
\]
The corresponding relation for the coefficients of the polynomials gives
\[
  \frac{\partial F_i}{\partial a_1}=M_1 \frac{\partial F_i}{\partial \kappa_1},
\]
as required.
\end{proof}
The proof of the \refP{main} is complete. Observe that \refL{commfact} can be generalized in the following way.
\begin{lemma}\label{L:commfact1}
Let $H$ be an arbitrary polynomial in $F_0,F_1,\ldots,F_{n-1}$.  Then
\[
  \frac{\partial H}{\partial a_k}=M_k\frac{\partial H}{\partial \kappa_k},\quad k=1,\ldots,n.
\]
\end{lemma}
\begin{proof}
\[
\frac{\partial H}{\partial a_k}=\sum_{i=0}^{n-1} \frac{\partial H}{\partial F_i} \frac{\partial F_i}{\partial a_k} = \sum_{i=0}^{n-1} \frac{\partial H}{\partial F_i}\left(M_k\frac{\partial F_i}{\partial \kappa_k}\right) = M_k\frac{\partial H}{\partial \kappa_k}.
\]
\end{proof}
\refL{commfact1} is the source of relations of type \refE{dynamics}. Indeed, by virtue of the canonical system of equations with the Hamiltonian $H$ we have $\dot a_k=\frac{\partial H}{\partial \kappa_k}$, $\dot\kappa_k=-\frac{\partial H}{\partial a_k}$, and \refL{commfact1} reads as
\[
   \dot\kappa_k=-M_k\dot a_k.
\]

\hoffset-1.0cm
\section{Genus 2. Codes}\label{S:Genus2}
{\small
(* As the first step {\bf we calculate the Lax operator:} *)
\begin{doublespace}
\noindent\(\pmb{ }\)
\end{doublespace}

\begin{doublespace}
\noindent\(\pmb{ }\\
\pmb{ A_0\text{:=}\left(
\begin{array}{cc}
 l_{11} & l_{12} \\
 l_{21} & l_{22}
\end{array}
\right)\text{      }  }
\pmb{A_1\text{:=}\left(
\begin{array}{cc}
 l_{111} & l_{112} \\
 l_{121} & l_{122}
\end{array}
\right)}
\text{  } 
\pmb{\alpha _1\text{:=}\left(
\begin{array}{c}
 \alpha _{11} \\
 \alpha _{21}
\end{array}
\right)}   \text{  }
\pmb{\alpha _2\text{:=}\left(
\begin{array}{c}
 \alpha _{12} \\
 \alpha _{22}
\end{array}
\right)}\text{  }
\pmb{\alpha _3\text{:=}\left(
\begin{array}{c}
 1 \\
 0
\end{array}
\right)}  \\  \text{  }
\pmb{ \alpha _4\text{:=}\left(
\begin{array}{c}
 0 \\
 1
\end{array}
\right)}\text{  }
\pmb{ \beta _1\text{:=}\left(
\begin{array}{c}
 \beta _{11} \\
 \beta _{21}
\end{array}
\right)}\text{  }
\pmb{\beta _2\text{:=}\left(
\begin{array}{c}
 \beta _{12} \\
 \beta _{22}
\end{array}
\right)}\text{  }
\pmb{\beta _3\text{:=}\left(
\begin{array}{c}
 0 \\
 \beta _{23}
\end{array}
\right)}\text{  }
\pmb{\beta _4\text{:=}\left(
\begin{array}{c}
 \beta _{14} \\
 0
\end{array}
\right)}\\ \\
\pmb{L\text{:=}A_0+A_1x+\alpha _1.\text{Transpose}\left[\beta _1\right]*\left(\left(y+b_1\right)/\left(x-a_1\right)\right)+\alpha _2.\text{Transpose}\left[\beta
_2\right]*\left(\left(y+b_2\right)/\left(x-a_2\right)\right)+}
\pmb{\alpha _3.\text{Transpose}\left[\beta _3\right]*\left(\left(y+b_3\right)/\left(x-a_3\right)\right)+\alpha _4.\text{Transpose}\left[\beta _4\right]*\left(\left(y+b_4\right)/\left(x-a_4\right)\right)}\\
\text{(*  }L \text{  is the Lax operator *)}
\)
\end{doublespace}

\begin{doublespace}
\noindent\(\pmb{ \text{MatrixForm}[L] }
\text{  (*   This results in the expression for $L$. Further on we form the {\bf eigenvalue conditions} *)}
\)
\end{doublespace}

\begin{doublespace}
\noindent\(\left(
\begin{array}{cc}
 l_{11}+\frac{l_{111}}{z^2}+\frac{\left(y+b_1\right) \alpha _{11} \beta _{11}}{\frac{1}{z^2}-a_1}+\frac{\left(y+b_2\right) \alpha _{12} \beta _{12}}{\frac{1}{z^2}-a_2}
& l_{12}+\frac{l_{112}}{z^2}+\frac{\left(y+b_1\right) \alpha _{11} \beta _{21}}{\frac{1}{z^2}-a_1}+\frac{\left(y+b_2\right) \alpha _{12} \beta _{22}}{\frac{1}{z^2}-a_2}+\frac{\left(y+b_3\right)
\beta _{23}}{\frac{1}{z^2}-a_3} \\
 l_{21}+\frac{l_{121}}{z^2}+\frac{\left(y+b_1\right) \alpha _{21} \beta _{11}}{\frac{1}{z^2}-a_1}+\frac{\left(y+b_2\right) \alpha _{22} \beta _{12}}{\frac{1}{z^2}-a_2}+\frac{\left(y+b_4\right)
\beta _{14}}{\frac{1}{z^2}-a_4} & l_{22}+\frac{l_{122}}{z^2}+\frac{\left(y+b_1\right) \alpha _{21} \beta _{21}}{\frac{1}{z^2}-a_1}+\frac{\left(y+b_2\right)
\alpha _{22} \beta _{22}}{\frac{1}{z^2}-a_2}
\end{array}
\right)\)
\end{doublespace}

\begin{doublespace}
\noindent\(\pmb{ \text{MatrixForm}[L\a_1] }
\text{   (* These 2 commands calculate the left hand *)}
\\
\pmb{ \text{MatrixForm}[L\a_2] }
\text{    (* sides of the subsequent system of equations *)}
\)
\end{doublespace}

\begin{doublespace}
\noindent\(\pmb{\text{sol}=}\\
\pmb{\text{Solve}[}\\
\pmb{\left\{l_{11}+a_3l_{111}+\frac{\left(b_3+b_1\right) \alpha _{11} \beta _{11}}{a_3-a_1}+\frac{\left(b_3+b_2\right) \alpha _{12} \beta _{12}}{a_3-a_2}\text{==}\kappa
_3\&\&l_{21}+a_3l_{121}+\frac{\left(b_3+b_1\right) \alpha _{21} \beta _{11}}{a_3-a_1}+\frac{\left(b_3+b_2\right) \alpha _{22} \beta _{12}}{a_3-a_2}+\frac{\left(b_3+b_4\right)
\beta _{14}}{a_3-a_4}\text{==}0\&\&\right.}\\
\pmb{l_{22}+a_4 l_{122}+\frac{\left(b_1+b_4\right) \alpha _{21} \beta _{21}}{-a_1+a_4}+\frac{\left(b_2+b_4\right) \alpha _{22} \beta _{22}}{-a_2+a_4}\text{==}\kappa
_4\&\&l_{12}+a_4 l_{112}+\frac{\left(b_1+b_4\right) \alpha _{11} \beta _{21}}{-a_1+a_4}+\frac{\left(b_2+b_4\right) \alpha _{12} \beta _{22}}{-a_2+a_4}+\frac{\left(b_3+b_4\right)
\beta _{23}}{-a_3+a_4}\text{==}0\&\&}\\
\pmb{\alpha _{11} \left(l_{11}+a_1 l_{111}+\frac{\left(b_1+b_2\right) \alpha _{12} \beta _{12}}{a_1-a_2}\right)+\alpha _{21} \left(l_{12}+a_1 l_{112}+\frac{\left(b_1+b_2\right)
\alpha _{12} \beta _{22}}{a_1-a_2}+\frac{\left(b_1+b_3\right) \beta _{23}}{a_1-a_3}\right)\text{==}\kappa _1\alpha _{11}\&\&}\\
\pmb{\alpha _{11} \left(l_{21}+a_1 l_{121}+\frac{\left(b_1+b_2\right) \alpha _{22} \beta _{12}}{a_1-a_2}+\frac{\left(b_1+b_4\right) \beta _{14}}{a_1-a_4}\right)+\alpha
_{21} \left(l_{22}+a_1 l_{122}+\frac{\left(b_1+b_2\right) \alpha _{22} \beta _{22}}{a_1-a_2}\right)\text{==}\kappa _1\alpha _{21}\&\&}\\
\pmb{\alpha _{12} \left(l_{11}+a_2 l_{111}+\frac{\left(b_1+b_2\right) \alpha _{11} \beta _{11}}{-a_1+a_2}\right)+\alpha _{22} \left(l_{12}+a_2 l_{112}+\frac{\left(b_1+b_2\right)
\alpha _{11} \beta _{21}}{-a_1+a_2}+\frac{\left(b_2+b_3\right) \beta _{23}}{a_2-a_3}\right)\text{==}\kappa _2\alpha _{12}\&\&}\\
\pmb{\left.\alpha _{12} \left(l_{21}+a_2 l_{121}+\frac{\left(b_1+b_2\right) \alpha _{21} \beta _{11}}{-a_1+a_2}+\frac{\left(b_2+b_4\right) \beta
_{14}}{a_2-a_4}\right)+\alpha _{22} \left(l_{22}+a_2 l_{122}+\frac{\left(b_1+b_2\right) \alpha _{21} \beta _{21}}{-a_1+a_2}\right)\text{==}\kappa
_2\alpha _{22}\right\},}\\
\pmb{\left.\left\{l_{11},l_{12},l_{21},l_{22},l_{111},l_{112},l_{121},l_{122}\right\}\right]}\)
\end{doublespace}

\begin{doublespace}
\noindent\(\pmb{\text{}}\\
\pmb{x\text{:=}1/z{}^{\wedge}2}\ \text{(* local parametrization at infinity *)}\\
\pmb{y=(1/z{}^{\wedge}5)\text{Sqrt}\left[1+p_1z{}^{\wedge}2+p_2z{}^{\wedge}4+p_3z{}^{\wedge}6+p_4z{}^{\wedge}8+p_5z{}^{\wedge}10\right]}\\
\pmb{\text{(*} \text{Calculating the Hamiltonians as the coefficients of the expansion of}\
T \text{at infinity}:\text{  }\text{*)}}\\
\pmb{T=\text{Series}[\text{Tr}[L.L],\{z,0,-2\}]}\\
\pmb{(-1/z{}^{\wedge}3)\text{Series}[1/y,\{z,0,9\}]\text{   }\text{(*} \text{dx}/y \text{ up to the factor dz}\text{*)}}\\
\pmb{\text{H2}=-\text{Residue}\left[(1/z)*T*\left(-z^2+\frac{p_1 z^4}{2}+\left(-\frac{3 p_1^2}{8}+\frac{p_2}{2}\right) z^6\right),\{z,0\}\right]
\text{(*} \text{with the differential } z{}^{\wedge}\{-1\}\text{dx}/y \text{*)}}\)
\end{doublespace}

\begin{doublespace}
\noindent\(\pmb{\left\{\left\{l_{11},l_{12},l_{21},l_{22},l_{111},l_{112},l_{121},l_{122}\right\}\right\}=\left\{l_{11},l_{12},l_{21},l_{22},l_{111},l_{112},l_{121},l_{122}\right\}\text{/.}\,\text{sol}}\)
\end{doublespace}

\begin{doublespace}
\noindent\pmb{\text{H2}}\ \ \text{(* by this line we obtain a final expression for H2 taking account of replacements *)}
\end{doublespace}

\noindent (*  The commands below correspond to the Hamiltonian equations with the Hamiltonian~H2. The assumptions correspond to the assumptions of the reduction. The zero results prove that assumptions of the reduction are preserved along the trajectories. D$\a$11 and so on denote the time derivatives of the corresponding variables.~*)

\begin{doublespace}
\noindent\(\pmb{\text{DA11}=\text{Series}[D[\text{H2},\beta _{11}]\text{/.}\,\{\alpha _{11}\to 0,\beta _{11}\to 0,\beta _{21}\to
0,\beta _{12}\to 0,\beta _{22}\to 0,\beta _{13}\to 0,\beta _{23}\to 0,\beta _{14}\to 0,}
\\
\pmb{\beta _{24}\to 0\},\{\alpha _{22},0,0\}]\text{/.}\,\{\alpha
_{22}\to 0\}}\)
\end{doublespace}

\begin{doublespace}
\noindent\(0\)
\end{doublespace}

\begin{doublespace}
\noindent\(\pmb{\text{D$\beta $11}=\text{Series}[D[\text{H2},\alpha _{11}]\text{/.}\,\{\alpha _{11}\to 0,\beta _{11}\to 0,\beta
_{21}\to 0,\beta _{12}\to 0,\beta _{22}\to 0,\beta _{13}\to 0,\beta _{23}\to 0,\beta _{14}\to 0,}
\\
\pmb{\beta _{24}\to 0\},\{\alpha _{22},0,0\}]\text{/.}\,\{\alpha
_{22}\to 0\} }\)
\end{doublespace}

\begin{doublespace}
\noindent\(0\)
\end{doublespace}

\begin{doublespace}
\noindent\(\pmb{\text{D$\beta $21}=\text{Series}[D[\text{H2},\alpha _{21}]\text{/.}\,\{\alpha _{11}\to 0,\beta _{11}\to 0,\beta
_{21}\to 0,\beta _{12}\to 0,\beta _{22}\to 0,\beta _{13}\to 0,\beta _{23}\to 0,\beta _{14}\to 0,}
\\
\pmb{\beta _{24}\to 0\},\{\alpha _{22},0,0\}]\text{/.}\,\{\alpha
_{22}\to 0\}}\)
\end{doublespace}

\begin{doublespace}
\noindent\(0\)
\end{doublespace}

\begin{doublespace}
\noindent\(\pmb{\text{D$\beta $12}=-\text{Series}[D[\text{H2},\alpha _{12}]\text{/.}\,\{\alpha _{11}\to 0,\beta _{11}\to 0,\beta
_{21}\to 0,\beta _{12}\to 0,\beta _{22}\to 0,\beta _{13}\to 0,\beta _{23}\to 0,\beta _{14}\to 0,}
\\
\pmb{\beta _{24}\to 0\},\{\alpha _{22},0,0\}]\text{/.}\,\{\alpha
_{22}\to 0\} }\)
\end{doublespace}

\begin{doublespace}
\noindent\(0\)
\end{doublespace}

\begin{doublespace}
\noindent\(\pmb{\text{D$\beta $22}=-\text{Series}[D[\text{H2},\alpha _{22}]\text{/.}\,\{\alpha _{11}\to 0,\beta _{11}\to 0,\beta
_{21}\to 0,\beta _{12}\to 0,\beta _{22}\to 0,\beta _{13}\to 0,\beta _{23}\to 0,\beta _{14}\to 0,}
\\
\pmb{\beta _{24}\to 0\},\{\alpha _{22},0,0\}]\text{/.}\,\{\alpha
_{22}\to 0\} }\)
\end{doublespace}

\begin{doublespace}
\noindent\(0\)
\end{doublespace}

\begin{doublespace}
\noindent\(\pmb{\text{D$\beta $13}=-\text{Series}[D[\text{H2},\alpha _{13}]\text{/.}\,\{\alpha _{11}\to 0,\beta _{11}\to 0,\beta
_{21}\to 0,\beta _{12}\to 0,\beta _{22}\to 0,\beta _{13}\to 0,\beta _{23}\to 0,\beta _{14}\to 0,}
\\
\pmb{\beta _{24}\to 0\},\{\alpha _{22},0,0\}]\text{/.}\,\{\alpha
_{22}\to 0\} }\)
\end{doublespace}

\begin{doublespace}
\noindent\(0\)
\end{doublespace}

\begin{doublespace}
\noindent\(\pmb{\text{D$\beta $23}=-\text{Series}[D[\text{H2},\alpha _{23}]\text{/.}\,\{\alpha _{11}\to 0,\beta _{11}\to 0,\beta
_{21}\to 0,\beta _{12}\to 0,\beta _{22}\to 0,\beta _{13}\to 0,\beta _{23}\to 0,\beta _{14}\to 0,}
\\
\pmb{\beta _{24}\to 0\},\{\alpha _{22},0,0\}]\text{/.}\,\{\alpha
_{22}\to 0\} }\)
\end{doublespace}

\begin{doublespace}
\noindent\(0\)
\end{doublespace}

\begin{doublespace}
\noindent\(\pmb{\text{D$\beta $14}=-\text{Series}[D[\text{H2},\alpha _{14}]\text{/.}\,\{\alpha _{11}\to 0,\beta _{11}\to 0,\beta
_{21}\to 0,\beta _{12}\to 0,\beta _{22}\to 0,\beta _{13}\to 0,\beta _{23}\to 0,\beta _{14}\to 0,}
\\
\pmb{\beta _{24}\to 0\},\{\alpha _{22},0,0\}]\text{/.}\,\{\alpha
_{22}\to 0\} }\)
\end{doublespace}

\begin{doublespace}
\noindent\(0\)
\end{doublespace}

\begin{doublespace}
\noindent\(\pmb{\text{D$\beta $24}=-\text{Series}[D[\text{H2},\alpha _{24}]\text{/.}\,\{\alpha _{11}\to 0,\beta _{11}\to 0,\beta
_{21}\to 0,\beta _{12}\to 0,\beta _{22}\to 0,\beta _{13}\to 0,\beta _{23}\to 0,\beta _{14}\to 0,}
\\
\pmb{\beta _{24}\to 0\},\{\alpha _{22},0,0\}]\text{/.}\,\{\alpha
_{22}\to 0\} }\)
\end{doublespace}

\begin{doublespace}
\noindent\(0\)
\end{doublespace}

\begin{doublespace}
\noindent\(\pmb{\text{D$\alpha $22}=\text{Series}[D[\text{H2},\beta _{22}]\text{/.}\,\{\alpha _{11}\to 0,\beta _{11}\to 0,\beta
_{21}\to 0,\beta _{12}\to 0,\beta _{22}\to 0,\beta _{13}\to 0,\beta _{23}\to 0,\beta _{14}\to 0,}
\\
\pmb{\beta _{24}\to 0\},\{\alpha _{22},0,0\}]\text{/.}\,\{\alpha
_{22}\to 0\}}\)
\end{doublespace}

\begin{doublespace}
\noindent\(0\)
\end{doublespace}

\begin{doublespace}
\noindent\(\pmb{\text{D$\alpha $12}=\text{Series}[D[\text{H2},\beta _{12}]\text{/.}\,\{\alpha _{11}\to 0,\beta _{11}\to 0,\beta
_{21}\to 0,\beta _{12}\to 0,\beta _{22}\to 0,\beta _{13}\to 0,\beta _{23}\to 0,\beta _{14}\to 0,}
\\
\pmb{\beta _{24}\to 0\},\{\alpha _{22},0,0\}]\text{/.}\,\{\alpha
_{22}\to 0\}}\)
\end{doublespace}

\begin{doublespace}
\noindent\(\frac{2 (b_2+b_3) (a_2 \alpha _{12} \kappa _2+a_3 \alpha _{12} \kappa _2-2 a_2 \alpha _{12} \kappa _3)}{(a_2-a_3){}^3}\)
\end{doublespace}

\begin{doublespace}
\noindent\(\pmb{\text{D$\alpha $21}=\text{Series}[D[\text{H2},\beta _{21}]\text{/.}\,\{\alpha _{11}\to 0,\beta _{11}\to 0,\beta
_{21}\to 0,\beta _{12}\to 0,\beta _{22}\to 0,\beta _{13}\to 0,\beta _{23}\to 0,\beta _{14}\to 0,}
\\
\pmb{\beta _{24}\to 0\},\{\alpha _{22},0,0\}]\text{/.}\,\{\alpha
_{22}\to 0\}}\)
\end{doublespace}

\begin{doublespace}
\noindent\(\frac{2 (b_1+b_4) (a_1 \alpha _{21} \kappa _1+a_4 \alpha _{21} \kappa _1-2 a_1 \alpha _{21} \kappa _4)}{(a_1-a_4){}^3}\)
\end{doublespace}

\begin{doublespace}
\noindent\(\pmb{\text{Da1}=\text{Series}[D[\text{H2},\kappa _1]\text{/.}\,\{\alpha _{11}\to 0,\beta _{11}\to 0,\beta _{21}\to
0,\beta _{12}\to 0,\beta _{22}\to 0,\beta _{13}\to 0,\beta _{23}\to 0,\beta _{14}\to 0,}
\\
\pmb{\beta _{24}\to 0\},\{\alpha _{22},0,0\}]\text{/.}\,\{\alpha
_{22}\to 0\}}\)
\end{doublespace}

\begin{doublespace}
\noindent\(\frac{2 (-2 a_4 \kappa _1+a_1 \kappa _4+a_4 \kappa _4)}{(a_1-a_4){}^2}\)
\end{doublespace}

\begin{doublespace}
\noindent\(\pmb{\text{Da2}=\text{Series}[D[\text{H2},\kappa _2]\text{/.}\,\{\alpha _{11}\to 0,\beta _{11}\to 0,\beta _{21}\to
0,\beta _{12}\to 0,\beta _{22}\to 0,\beta _{13}\to 0,\beta _{23}\to 0,\beta _{14}\to 0,}
\\
\pmb{\beta _{24}\to 0\},\{\alpha _{22},0,0\}]\text{/.}\,\{\alpha
_{22}\to 0\}}\)
\end{doublespace}

\begin{doublespace}
\noindent\(\frac{2 (-2 a_3 \kappa _2+a_2 \kappa _3+a_3 \kappa _3)}{(a_2-a_3){}^2}\)
\end{doublespace}

\begin{doublespace}
\noindent\(\pmb{\text{Da3}=\text{Series}[D[\text{H2},\kappa _3]\text{/.}\,\{\alpha _{11}\to 0,\beta _{11}\to 0,\beta _{21}\to
0,\beta _{12}\to 0,\beta _{22}\to 0,\beta _{13}\to 0,\beta _{23}\to 0,\beta _{14}\to 0,}
\\
\pmb{\beta _{24}\to 0\},\{\alpha _{22},0,0\}]\text{/.}\,\{\alpha
_{22}\to 0\}}\)
\end{doublespace}

\begin{doublespace}
\noindent\(\frac{2 (a_2 \kappa _2+a_3 \kappa _2-2 a_2 \kappa _3) }{(a_2-a_3){}^2}\)
\end{doublespace}

\begin{doublespace}
\noindent\(\pmb{\text{Da4}=\text{Series}[D[\text{H2},\kappa _4]\text{/.}\,\{\alpha _{11}\to 0,\beta _{11}\to 0,\beta _{21}\to
0,\beta _{12}\to 0,\beta _{22}\to 0,\beta _{13}\to 0,\beta _{23}\to 0,\beta _{14}\to 0,}
\\
\pmb{\beta _{24}\to 0\},\{\alpha _{22},0,0\}]\text{/.}\,\{\alpha
_{22}\to 0\}}\)
\end{doublespace}

\begin{doublespace}
\noindent\(\frac{2 (a_1 \kappa _1+a_4 \kappa _1-2 a_1 \kappa _4) }{(a_1-a_4){}^2}\)
\end{doublespace}

\begin{doublespace}
\noindent\(\pmb{\text{D$\kappa $1}=\text{Series}[-D[\text{H2},a_1]\text{/.}\,\{\alpha _{11}\to 0,\beta _{11}\to 0,\beta _{21}\to
0,\beta _{12}\to 0,\beta _{22}\to 0,\beta _{13}\to 0,\beta _{23}\to 0,\beta _{14}\to 0,}
\\
\pmb{\beta _{24}\to 0\},\{\alpha _{22},0,0\}]\text{/.}\,\{\alpha
_{22}\to 0\}}\)
\end{doublespace}

\begin{doublespace}
\noindent\(\frac{2 (\kappa _1-\kappa _4) (-2 a_4 \kappa _1+a_1 \kappa _4+a_4 \kappa _4)}{(a_1-a_4){}^3}\)
\end{doublespace}

\begin{doublespace}
\noindent\(\pmb{\text{D$\kappa $2}=\text{Series}[-D[\text{H2},a_2]\text{/.}\,\{\alpha _{11}\to 0,\beta _{11}\to 0,\beta _{21}\to
0,\beta _{12}\to 0,\beta _{22}\to 0,\beta _{13}\to 0,\beta _{23}\to 0,\beta _{14}\to 0,}
\\
\pmb{\beta _{24}\to 0\},\{\alpha _{22},0,0\}]\text{/.}\,\{\alpha
_{22}\to 0\}}\)
\end{doublespace}

\begin{doublespace}
\noindent\(\frac{2 (\kappa _2-\kappa _3) (-2 a_3 \kappa _2+a_2 \kappa _3+a_3 \kappa _3)}{(a_2-a_3){}^3}\)
\end{doublespace}

\begin{doublespace}
\noindent\(\pmb{\text{D$\kappa $3}=\text{Series}[-D[\text{H2},a_3]\text{/.}\,\{\alpha _{11}\to 0,\beta _{11}\to 0,\beta _{21}\to
0,\beta _{12}\to 0,\beta _{22}\to 0,\beta _{13}\to 0,\beta _{23}\to 0,\beta _{14}\to 0,}
\\
\pmb{\beta _{24}\to 0\},\{\alpha _{22},0,0\}]\text{/.}\,\{\alpha
_{22}\to 0\}}\)
\end{doublespace}

\begin{doublespace}
\noindent\(\frac{2 (\kappa _2-\kappa _3) (a_2 \kappa _2+a_3 \kappa _2-2 a_2 \kappa _3)}{(a_2-a_3){}^3}\)
\end{doublespace}

\begin{doublespace}
\noindent\(\pmb{\text{D$\kappa $4}=\text{Series}[-D[\text{H2},a_4]\text{/.}\,\{\alpha _{11}\to 0,\beta _{11}\to 0,\beta _{21}\to
0,\beta _{12}\to 0,\beta _{22}\to 0,\beta _{13}\to 0,\beta _{23}\to 0,\beta _{14}\to 0,}
\\
\pmb{\beta _{24}\to 0\},\{\alpha _{22},0,0\}]\text{/.}\,\{\alpha
_{22}\to 0\}}\)
\end{doublespace}

\begin{doublespace}
\noindent\(-\frac{2 (\kappa _1-\kappa _4) (a_1 \kappa _1+a_4 \kappa _1-2 a_1 \kappa _4)}{(-a_1+a_4){}^3}\)
\end{doublespace}

\begin{doublespace}
\noindent\(\pmb{\text{H2r}=\text{Simplify}[\text{Series}[\text{H2}\text{/.}\,\{\alpha _{11}\to 0,\beta _{11}\to 0,\beta _{21}\to 0,\beta
_{12}\to 0,\beta _{22}\to 0,\beta _{13}\to 0,\beta _{23}\to 0,\beta _{14}\to 0,}
\\
\pmb{\beta _{24}\to 0\},\{\alpha _{22},0,0\}]\text{/.}\,\{\alpha
_{22}\to 0\}]}\)
\end{doublespace}

\begin{doublespace}
\noindent\(-\frac{1}{(a_2-a_3){}^2 (a_1-a_4){}^2}2 (a_2^2 (\kappa _1-\kappa _4) (a_4 \kappa _1-a_1 \kappa
_4)+a_2 (-(a_1-a_4){}^2 (\kappa _2-\kappa _3) \kappa _3-2 a_3 (\kappa _1-\kappa _4) (a_4 \kappa
_1-a_1 \kappa _4))+a_3 ((a_1-a_4){}^2 \kappa _2 (\kappa _2-\kappa _3)+a_3 (\kappa _1-\kappa _4)
(a_4 \kappa _1-a_1 \kappa _4)))\)
\end{doublespace}

\begin{doublespace}
\noindent\(\pmb{\text{(*} \text{It is equal to}\ 2(\kappa _1-\kappa _4)(a_4\kappa _1-a_1\kappa _4)/(a_1-a_4){}^2+
2(\kappa _2-\kappa _3)(a_3\kappa _2-a_2\kappa _3)/(a_2-a_3){}^2\ \text{as in  Proposition 1}\text{  }\text{*)}}\)
\end{doublespace}
}
\section{Genus 3. Codes}\label{S:Genus3}
{\small
\begin{doublespace}
\noindent\(\pmb{ }\\
\pmb{ A_0\text{:=}\left(
\begin{array}{cc}
 l_{11} & l_{12} \\
 l_{21} & l_{22}
\end{array}
\right)\text{      }A_1\text{:=}\left(
\begin{array}{cc}
 l_{111} & l_{112} \\
 l_{121} & l_{122}
\end{array}
\right)\text{   }A_2\text{:=}\left(
\begin{array}{cc}
 l_{211} & l_{212} \\
 l_{221} & l_{222}
\end{array}
\right) }\\
\pmb{\alpha _1\text{:=}\left(
\begin{array}{c}
 \alpha _{11} \\
 \alpha _{21}
\end{array}
\right)\text{  }\alpha _2\text{:=}\left(
\begin{array}{c}
 \alpha _{12} \\
 \alpha _{22}
\end{array}
\right)\text{   }\alpha _3\text{:=}\left(
\begin{array}{c}
 \alpha _{13} \\
 \alpha _{23}
\end{array}
\right)\text{  }\alpha _4\text{:=}\left(
\begin{array}{c}
 \alpha _{14} \\
 \alpha _{24}
\end{array}
\right)\text{  }\alpha _5\text{:=}\left(
\begin{array}{c}
 1 \\
 0
\end{array}
\right)\text{    }\alpha _6\text{:=}\left(
\begin{array}{c}
 0 \\
 1
\end{array}
\right)}\\
\pmb{ \beta _1\text{:=}\left(
\begin{array}{c}
 \beta _{11} \\
 \beta _{21}
\end{array}
\right)\text{  }\beta _2\text{:=}\left(
\begin{array}{c}
 \beta _{12} \\
 \beta _{22}
\end{array}
\right)\text{  }\beta _3\text{:=}\left(
\begin{array}{c}
 \beta _{13} \\
 \beta _{23}
\end{array}
\right)\text{  }\beta _4\text{:=}\left(
\begin{array}{c}
 \beta _{14} \\
 \beta _{24}
\end{array}
\right)\text{  }\beta _5\text{:=}\left(
\begin{array}{c}
 0 \\
 \beta _{25}
\end{array}
\right)\text{  }\beta _6\text{:=}\left(
\begin{array}{c}
 \beta _{16} \\
 0
\end{array}
\right)}\\
\pmb{\text{(*}\ L\ \text{is the Lax operator:}
\text{  *)}}\\
\pmb{L\text{:=}A_0+A_1x+A_2x^2+\alpha _1.\text{Transpose}[\beta _1]*((y+b_1)/(x-a_1))+\alpha _2.\text{Transpose}[\beta
_2]*((y+b_2)/(x-a_2))+}
\pmb{\alpha _3.\text{Transpose}[\beta _3]*((y+b_3)/(x-a_3))+\alpha _4.\text{Transpose}[\beta _4]*((y+b_4)/(x-a_4))+\alpha
_5.\text{Transpose}[\beta _5]*((y+b_5)/(x-a_5))+}\\
\pmb{\alpha _6.\text{Transpose}[\beta _6]*((y+b_6)/(x-a_6))}\)
\end{doublespace}

\begin{doublespace}
\noindent\(\pmb{\text{(*} \text{Calculating the Hamiltonian}: \text{*)}}\)
\end{doublespace}

\begin{doublespace}
\noindent\(\pmb{x\text{:=}1/z{}^{\wedge}2}\)
\end{doublespace}

\begin{doublespace}
\noindent\(\pmb{y=(1/z{}^{\wedge}7)\text{Sqrt}[1+p_1z{}^{\wedge}2+p_2z{}^{\wedge}4+p_3z{}^{\wedge}6+p_4z{}^{\wedge}8+p_5z{}^{\wedge}10+p_6z{}^{\wedge}12+p_7z{}^{\wedge}14]}\)
\end{doublespace}

\begin{doublespace}
\noindent\(\pmb{\text{(* at the infinity}\ 1/y\sim z^7,\text{  }z^{-1}\text{dx}/y\sim z^3\text{dz},\ \text{hence}\ \text{Tr}[L.L]\
\text{should be calculated up to}\ z^{-4}\ \text{  *)}}\)
\end{doublespace}

\begin{doublespace}
\noindent\(\pmb{T=\text{Series}[\text{Tr}[L.L],\{z,0,-4\}]}\)
\end{doublespace}

\begin{doublespace}
\noindent\(\pmb{\text{(*} \text{Tr}[L.L]\sim z^{-10}, z^{-1}\text{Tr}[L.L]\sim z^{-11},\ \text{hence}\ \text{dx}/y\ \text{should be calculated up to}\ z^{10}\ \text{*)}}\)
\end{doublespace}

\begin{doublespace}
\noindent\(\pmb{(-1/z{}^{\wedge}3)\text{Series}[1/y,\{z,0,13\}]\text{   }\text{(*}\ \text{dx}/y\ \text{up to the factor}\ \text{dz}\text{ *)}}\)
\end{doublespace}

\begin{doublespace}
\noindent\(\pmb{\text{(*} \text{For}\ p_1=p_2=p_3=0\ \text{the Hamiltonian at  } z^{-5}\ \text{gives trivial equations. *)}}\)
\end{doublespace}

\begin{doublespace}
\noindent\(\pmb{\text{(*}\text{  }\text{The following Hamiltonian is the Hamiltonian at }\ z^{-4}\
(\text{for}\ p_1=p_2=p_3=0):\ \text{*)}}\)
\end{doublespace}

\begin{doublespace}
\noindent\(\pmb{\text{H2}=-\text{Residue}[(1/z)*T*(-z^4+\frac{p_1 z^6}{2}+(-\frac{3 p_1^2}{8}+\frac{p_2}{2}) z^8+\frac{1}{16}
(5 p_1^3-12 p_1 p_2+8 p_3) z^{10}),\{z,0\}]\text{    }}\\
\pmb{\text{(* with the differential}\ z^{-1}\text{dx}/y.\ \text{Below}\ p_1=p_2=p_3=0 \text{ *)}}\)
\end{doublespace}

\begin{doublespace}
\noindent\(\pmb{\text{(* First part of the equations out of eigenvalue conditions} }\\
\pmb{(\text{to generate the left hand sides of those equations one makes use of commands of the type MatrixForm}[\text{L$\alpha $}_i] }\\
\pmb{\text{like it has been done for genus}=2):\text{  }\text{*)}}\)
\end{doublespace}

\begin{doublespace}
\noindent\(\pmb{\text{sol}=}\\
\pmb{\text{Solve}[\{l_{11}+a_5l_{111}+a_5{}^2l_{211}+\frac{(b_5+b_1) \alpha _{11} \beta _{11}}{a_5-a_1}+\frac{(b_5+b_2)
\alpha _{12} \beta _{12}}{a_5-a_2}+\frac{(b_5+b_3) \alpha _{13} \beta _{13}}{a_5-a_3}+\frac{(b_5+b_4) \alpha _{14} \beta _{14}}{a_5-a_4}\text{==}\kappa
_5\&\&..}\\
\pmb{l_{12}+a_6 l_{112}+a_6{}^2l_{212}+\frac{(b_1+b_6) \alpha _{11} \beta _{21}}{-a_1+a_6}+\frac{(b_2+b_6) \alpha _{12} \beta
_{22}}{-a_2+a_6}+\frac{(b_3+b_6) \alpha _{13} \beta _{23}}{-a_3+a_6}+\frac{(b_4+b_6) \alpha _{14} \beta _{24}}{-a_4+a_6}+\frac{(b_5+b_6)
\beta _{25}}{-a_5+a_6}\text{==}0\&\&}\\
\pmb{\alpha _{11} (l_{11}+a_1 l_{111}+a_1{}^2l_{211}+\frac{(b_1+b_2) \alpha _{12} \beta _{12}}{a_1-a_2}+\frac{(b_1+b_3)
\alpha _{13} \beta _{13}}{a_1-a_3}+\frac{(b_1+b_4) \alpha _{14} \beta _{14}}{a_1-a_4})+}\\
\pmb{\alpha _{21} (l_{12}+a_1 l_{112}+a_1{}^2l_{212}+\frac{(b_1+b_2) \alpha _{12} \beta _{22}}{a_1-a_2}+\frac{(b_1+b_3)
\alpha _{13} \beta _{23}}{a_1-a_3}+\frac{(b_1+b_4) \alpha _{14} \beta _{24}}{a_1-a_4}+\frac{(b_1+b_5) \beta _{25}}{a_1-a_5})\text{==}\kappa
_1\alpha _{11}\&\&}\\
\pmb{\alpha _{12} (l_{11}+a_2 l_{111}+a_2{}^2l_{211}+\frac{(b_1+b_2) \alpha _{11} \beta _{11}}{-a_1+a_2}+\frac{(b_3+b_2)
\alpha _{13} \beta _{13}}{-a_3+a_2}+\frac{(b_4+b_2) \alpha _{14} \beta _{14}}{-a_4+a_2})+}\\
\pmb{\alpha _{22} (l_{12}+a_2 l_{112}+a_2{}^2l_{212}+\frac{(b_1+b_2) \alpha _{11} \beta _{21}}{-a_1+a_2}+\frac{(b_3+b_2)
\alpha _{13} \beta _{23}}{-a_3+a_2}+\frac{(b_4+b_2) \alpha _{14} \beta _{24}}{-a_4+a_2}+\frac{(b_2+b_5) \beta _{25}}{a_2-a_5})\text{==}\kappa
_2\alpha _{12}\&\&}\\
\pmb{\alpha _{13} (l_{11}+a_3 l_{111}+a_3{}^2l_{211}+\frac{(b_1+b_3) \alpha _{11} \beta _{11}}{-a_1+a_3}+\frac{(b_3+b_2)
\alpha _{12} \beta _{12}}{-a_2+a_3}+\frac{(b_4+b_3) \alpha _{14} \beta _{14}}{-a_4+a_3})+}\\
\pmb{\alpha _{23} (l_{12}+a_3 l_{112}+a_3{}^2l_{212}+\frac{(b_1+b_3) \alpha _{11} \beta _{21}}{-a_1+a_3}+\frac{(b_3+b_2)
\alpha _{12} \beta _{22}}{-a_2+a_3}+\frac{(b_4+b_3) \alpha _{14} \beta _{24}}{-a_4+a_3}+\frac{(b_3+b_5) \beta _{25}}{a_3-a_5})\text{==}\kappa
_3\alpha _{13}\&\&}\\
\pmb{\alpha _{14} (l_{11}+a_4 l_{111}+a_4{}^2l_{211}+\frac{(b_1+b_4) \alpha _{11} \beta _{11}}{-a_1+a_4}+\frac{(b_4+b_2)
\alpha _{12} \beta _{12}}{-a_2+a_4}+\frac{(b_4+b_3) \alpha _{13} \beta _{13}}{-a_3+a_4})+}\\
\pmb{.\alpha _{24} (l_{12}+a_4 l_{112}+a_4{}^2l_{212}+\frac{(b_1+b_4) \alpha _{11} \beta _{21}}{-a_1+a_4}+\frac{(b_4+b_2)
\alpha _{12} \beta _{22}}{-a_2+a_4}+\frac{(b_4+b_3) \alpha _{13} \beta _{23}}{-a_3+a_4}+\frac{(b_4+b_5) \beta _{25}}{a_4-a_5})\text{==}\kappa
_4\alpha _{14}\},}\\
\pmb{.\{l_{11},l_{12},l_{111},l_{112},l_{211},l_{212}\}]}\)
\end{doublespace}

\begin{doublespace}
\noindent\(\pmb{\{\{l_{11},l_{12},l_{111},l_{112},l_{211},l_{212}\}\}=\{l_{11},l_{12},l_{111},l_{112},l_{211},l_{212}\}\text{/.}\text{sol}}\)
\end{doublespace}

\begin{doublespace}
\noindent\(\pmb{\text{(* Second part of equations out of eigenvalue conditions:    *)}}\)
\end{doublespace}

\begin{doublespace}
\noindent\(\pmb{\text{sol}=}\\
\pmb{\text{Solve}[\{l_{21}+a_5l_{121}+a_5{}^2l_{221}+\frac{(b_5+b_1) \alpha _{21} \beta _{11}}{a_5-a_1}+\frac{(b_5+b_2)
\alpha _{22} \beta _{12}}{a_5-a_2}+\frac{(b_5+b_3) \alpha _{23} \beta _{13}}{a_5-a_3}+\frac{(b_5+b_4) \alpha _{24} \beta _{14}}{a_5-a_4}+\frac{(b_5+b_6)
\beta _{16}}{a_5-a_6}\text{==}0\&\&..}\\
\pmb{l_{22}+a_6 l_{122}+a_6{}^2l_{222}+\frac{(b_1+b_6) \alpha _{21} \beta _{21}}{-a_1+a_6}+\frac{(b_2+b_6) \alpha _{22} \beta
_{22}}{-a_2+a_6}+\frac{(b_3+b_6) \alpha _{23} \beta _{23}}{-a_3+a_6}+\frac{(b_4+b_6) \alpha _{24} \beta _{24}}{-a_4+a_6}\text{==}\kappa
_6\&\&}\\
\pmb{\alpha _{11} (l_{21}+a_1 l_{121}+a_1{}^2l_{221}+\frac{(b_1+b_2) \alpha _{22} \beta _{12}}{a_1-a_2}+\frac{(b_1+b_3)
\alpha _{23} \beta _{13}}{a_1-a_3}+\frac{(b_1+b_4) \alpha _{24} \beta _{14}}{a_1-a_4}+\frac{(b_1+b_6) \beta _{16}}{a_1-a_6})+}\\
\pmb{\alpha _{21} (l_{22}+a_1 l_{122}+a_1{}^2l_{222}+\frac{(b_1+b_2) \alpha _{22} \beta _{22}}{a_1-a_2}+\frac{(b_1+b_3)
\alpha _{23} \beta _{23}}{a_1-a_3}+\frac{(b_1+b_4) \alpha _{24} \beta _{24}}{a_1-a_4})\text{==}\kappa _1\alpha _{21}\&\&}\\
\pmb{\alpha _{12} (l_{21}+a_2 l_{121}+a_2{}^2l_{221}+\frac{(b_1+b_2) \alpha _{21} \beta _{11}}{-a_1+a_2}+\frac{(b_3+b_2)
\alpha _{23} \beta _{13}}{-a_3+a_2}+\frac{(b_4+b_2) \alpha _{24} \beta _{14}}{-a_4+a_2}+\frac{(b_2+b_6) \beta _{16}}{a_2-a_6})+}\\
\pmb{\alpha _{22} (l_{22}+a_2 l_{122}+a_2{}^2l_{222}+\frac{(b_1+b_2) \alpha _{21} \beta _{21}}{-a_1+a_2}+\frac{(b_3+b_2)
\alpha _{23} \beta _{23}}{-a_3+a_2}+\frac{(b_4+b_2) \alpha _{24} \beta _{24}}{-a_4+a_2})\text{==}\kappa _2\alpha _{22}\&\&}\\
\pmb{\alpha _{13} (l_{21}+a_3 l_{121}+a_3{}^2l_{221}+\frac{(b_1+b_3) \alpha _{21} \beta _{11}}{-a_1+a_3}+\frac{(b_3+b_2)
\alpha _{22} \beta _{12}}{-a_2+a_3}+\frac{(b_4+b_3) \alpha _{24} \beta _{14}}{-a_4+a_3}+\frac{(b_3+b_6) \beta _{16}}{a_3-a_6})+}\\
\pmb{\alpha _{23} (l_{22}+a_3 l_{122}+a_3{}^2l_{222}+\frac{(b_1+b_3) \alpha _{21} \beta _{21}}{-a_1+a_3}+\frac{(b_3+b_2)
\alpha _{22} \beta _{22}}{-a_2+a_3}+\frac{(b_4+b_3) \alpha _{24} \beta _{24}}{-a_4+a_3})\text{==}\kappa _3\alpha _{23}\&\&}\\
\pmb{\alpha _{14} (l_{21}+a_4 l_{121}+a_4{}^2l_{221}+\frac{(b_1+b_4) \alpha _{21} \beta _{11}}{-a_1+a_4}+\frac{(b_4+b_2)
\alpha _{22} \beta _{12}}{-a_2+a_4}+\frac{(b_4+b_3) \alpha _{23} \beta _{13}}{-a_3+a_4}+\frac{(b_4+b_6) \beta _{16}}{a_4-a_6})+}\\
\pmb{..\alpha _{24} (l_{22}+a_4 l_{122}+a_4{}^2l_{222}+\frac{(b_1+b_4) \alpha _{21} \beta _{21}}{-a_1+a_4}+\frac{(b_4+b_2)
\alpha _{22} \beta _{22}}{-a_2+a_4}+\frac{(b_4+b_3) \alpha _{23} \beta _{23}}{-a_3+a_4})\text{==}\kappa _4\alpha _{24}\},\{l_{21},l_{22},l_{121},l_{122},l_{221},l_{222}\}]}\)
\end{doublespace}

\begin{doublespace}
\noindent\(\pmb{\{\{l_{21},l_{22},l_{121},l_{122},l_{221},l_{222}\}\}=\{l_{21},l_{22},l_{121},l_{122},l_{221},l_{222}\}\text{/.}\text{sol}}\)
\end{doublespace}

\begin{doublespace}
\noindent\(\pmb{\text{(*  It is recommended to collect and keep the results of the foregoing replacements for}\ l_{\text{ij}},}\\
\pmb{l_{\text{ijk}}\ \text{in a separate file because solution of the system takes quite
a lot time *)}}\)
\end{doublespace}

\begin{doublespace}
\noindent\(\pmb{\text{H2}\text{  }\text{(*} \text{calculating the Hamiltonian taking account of the replacements for}\ l_{\text{ij}},l_{\text{ijk}}\text{  *)}}\)
\end{doublespace}
(* Next we check admissibility of the reduction:  *)

\begin{doublespace}
\noindent\(\pmb{\text{DA11}=}
\pmb{\text{Simplify}[}
\pmb{\text{Series}[D[\text{H2},\beta _{11}]\text{/.}\,\{\alpha _{11}\to 0,\alpha _{13}\to 0,\alpha _{24}\to 0,\beta _{11}\to 0,\beta
_{21}\to 0,\beta _{12}\to 0,\beta _{22}\to 0,\beta _{13}\to 0,}\\
\pmb{\beta _{23}\to 0,\beta _{14}\to 0,\beta _{24}\to 0,\beta _{25}\to 0,\beta _{16}\to
0\},\{\alpha _{22},0,0\}]\text{/.}\,\{\alpha _{22}\to 0\}]}\)
\end{doublespace}

\begin{doublespace}
\noindent\(0\)
\end{doublespace}

\begin{doublespace}
\noindent\(\pmb{\text{DA22}=}
\pmb{\text{Simplify}[}
\pmb{\text{Series}[D[\text{H2},\beta _{22}]\text{/.}\,\{\alpha _{11}\to 0,\alpha _{13}\to 0,\alpha _{24}\to 0,\beta _{11}\to 0,\beta
_{21}\to 0,\beta _{12}\to 0,\beta _{22}\to 0,\beta _{13}\to 0,}\\
\pmb{\beta _{23}\to 0,\beta _{14}\to 0,\beta _{24}\to 0,\beta _{25}\to 0,\beta _{16}\to
0\},\{\alpha _{22},0,0\}]\text{/.}\,\{\alpha _{22}\to 0\}]}\)
\end{doublespace}

\begin{doublespace}
\noindent\(0\)
\end{doublespace}

\begin{doublespace}
\noindent\(\pmb{\text{DA13}=}
\pmb{\text{Simplify}[}
\pmb{\text{Series}[D[\text{H2},\beta _{13}]\text{/.}\,\{\alpha _{11}\to 0,\alpha _{13}\to 0,\alpha _{24}\to 0,\beta _{11}\to 0,\beta
_{21}\to 0,\beta _{12}\to 0,\beta _{22}\to 0,\beta _{13}\to 0,}\\
\pmb{\beta _{23}\to 0,\beta _{14}\to 0,\beta _{24}\to 0,\beta _{25}\to 0,\beta _{16}\to
0\},\{\alpha _{22},0,0\}]\text{/.}\,\{\alpha _{22}\to 0\}]}\)
\end{doublespace}

\begin{doublespace}
\noindent\(0\)
\end{doublespace}

\begin{doublespace}
\noindent\(\pmb{\text{DA24}=}
\pmb{\text{Simplify}[}
\pmb{\text{Series}[D[\text{H2},\beta _{24}]\text{/.}\,\{\alpha _{11}\to 0,\alpha _{13}\to 0,\alpha _{24}\to 0,\beta _{11}\to 0,\beta
_{21}\to 0,\beta _{12}\to 0,\beta _{22}\to 0,\beta _{13}\to 0,}\\
\pmb{\beta _{23}\to 0,\beta _{14}\to 0,\beta _{24}\to 0,\beta _{25}\to 0,\beta _{16}\to
0\},\{\alpha _{22},0,0\}]\text{/.}\,\{\alpha _{22}\to 0\}] }\)
\end{doublespace}

\begin{doublespace}
\noindent\(0\)
\end{doublespace}

\begin{doublespace}
\noindent\(\pmb{\text{DB11}=}
\pmb{\text{Simplify}[}
\pmb{\text{Series}[D[\text{H2},\alpha _{11}]\text{/.}\,\{\alpha _{11}\to 0,\alpha _{13}\to 0,\alpha _{24}\to 0,\beta _{11}\to
0,\beta _{21}\to 0,\beta _{12}\to 0,\beta _{22}\to 0,\beta _{13}\to 0,}\\
\pmb{\beta _{23}\to 0,\beta _{14}\to 0,\beta _{24}\to 0,\beta _{25}\to 0,\beta _{16}\to
0\},\{\alpha _{22},0,0\}]\text{/.}\,\{\alpha _{22}\to 0\}] }\)
\end{doublespace}

\begin{doublespace}
\noindent\(0\)
\end{doublespace}

\begin{doublespace}
\noindent\(\pmb{\text{DB21}=}
\pmb{\text{Simplify}[}
\pmb{\text{Series}[D[\text{H2},\alpha _{21}]\text{/.}\,\{\alpha _{11}\to 0,\alpha _{13}\to 0,\alpha _{24}\to 0,\beta _{11}\to
0,\beta _{21}\to 0,\beta _{12}\to 0,\beta _{22}\to 0,\beta _{13}\to 0,}\\
\pmb{\beta _{23}\to 0,\beta _{14}\to 0,\beta _{24}\to 0,\beta _{25}\to 0,\beta _{16}\to
0\},\{\alpha _{22},0,0\}]\text{/.}\,\{\alpha _{22}\to 0\}] }\)
\end{doublespace}

\begin{doublespace}
\noindent\(0\)
\end{doublespace}

\begin{doublespace}
\noindent\(\pmb{\text{DB12}=}
\pmb{\text{Simplify}[}
\pmb{\text{Series}[D[\text{H2},\alpha _{12}]\text{/.}\,\{\alpha _{11}\to 0,\alpha _{13}\to 0,\alpha _{24}\to 0,\beta _{11}\to
0,\beta _{21}\to 0,\beta _{12}\to 0,\beta _{22}\to 0,\beta _{13}\to 0,}\\
\pmb{\beta _{23}\to 0,\beta _{14}\to 0,\beta _{24}\to 0,\beta _{25}\to 0,\beta _{16}\to
0\},\{\alpha _{22},0,0\}]\text{/.}\,\{\alpha _{22}\to 0\}] }\)
\end{doublespace}

\begin{doublespace}
\noindent\(0\)
\end{doublespace}

\begin{doublespace}
\noindent\(\pmb{\text{DB22}=}
\pmb{\text{Simplify}[}
\pmb{\text{Series}[D[\text{H2},\alpha _{22}]\text{/.}\,\{\alpha _{11}\to 0,\alpha _{13}\to 0,\alpha _{24}\to 0,\beta _{11}\to
0,\beta _{21}\to 0,\beta _{12}\to 0,\beta _{22}\to 0,\beta _{13}\to 0,}\\
\pmb{\beta _{23}\to 0,\beta _{14}\to 0,\beta _{24}\to 0,\beta _{25}\to 0,\beta _{16}\to
0\},\{\alpha _{22},0,0\}]\text{/.}\,\{\alpha _{22}\to 0\}] }\)
\end{doublespace}

\begin{doublespace}
\noindent\(0\)
\end{doublespace}

\begin{doublespace}
\noindent\(\pmb{\text{DB13}=}
\pmb{\text{Simplify}[}
\pmb{\text{Series}[D[\text{H2},\alpha _{13}]\text{/.}\,\{\alpha _{11}\to 0,\alpha _{13}\to 0,\alpha _{24}\to 0,\beta _{11}\to
0,\beta _{21}\to 0,\beta _{12}\to 0,\beta _{22}\to 0,\beta _{13}\to 0,}\\
\pmb{\beta _{23}\to 0,\beta _{14}\to 0,\beta _{24}\to 0,\beta _{25}\to 0,\beta _{16}\to
0\},\{\alpha _{22},0,0\}]\text{/.}\,\{\alpha _{22}\to 0\}] }\)
\end{doublespace}

\begin{doublespace}
\noindent\(0\)
\end{doublespace}

\begin{doublespace}
\noindent\(\pmb{\text{DB23}=}
\pmb{\text{Simplify}[}
\pmb{\text{Series}[D[\text{H2},\alpha _{23}]\text{/.}\,\{\alpha _{11}\to 0,\alpha _{13}\to 0,\alpha _{24}\to 0,\beta _{11}\to
0,\beta _{21}\to 0,\beta _{12}\to 0,\beta _{22}\to 0,\beta _{13}\to 0,}\\
\pmb{\beta _{23}\to 0,\beta _{14}\to 0,\beta _{24}\to 0,\beta _{25}\to 0,\beta _{16}\to
0\},\{\alpha _{22},0,0\}]\text{/.}\,\{\alpha _{22}\to 0\}] }\)
\end{doublespace}

\begin{doublespace}
\noindent\(0\)
\end{doublespace}

\begin{doublespace}
\noindent\(\pmb{\text{DB14}=}
\pmb{\text{Simplify}[}
\pmb{\text{Series}[D[\text{H2},\alpha _{14}]\text{/.}\,\{\alpha _{11}\to 0,\alpha _{13}\to 0,\alpha _{24}\to 0,\beta _{11}\to
0,\beta _{21}\to 0,\beta _{12}\to 0,\beta _{22}\to 0,\beta _{13}\to 0,}\\
\pmb{\beta _{23}\to 0,\beta _{14}\to 0,\beta _{24}\to 0,\beta _{25}\to 0,\beta _{16}\to
0\},\{\alpha _{22},0,0\}]\text{/.}\,\{\alpha _{22}\to 0\}] }\)
\end{doublespace}

\begin{doublespace}
\noindent\(0\)
\end{doublespace}

\begin{doublespace}
\noindent\(\pmb{\text{DB24}=}
\pmb{\text{Simplify}[}
\pmb{\text{Series}[D[\text{H2},\alpha _{24}]\text{/.}\,\{\alpha _{11}\to 0,\alpha _{13}\to 0,\alpha _{24}\to 0,\beta _{11}\to
0,\beta _{21}\to 0,\beta _{12}\to 0,\beta _{22}\to 0,\beta _{13}\to 0,}\\
\pmb{\beta _{23}\to 0,\beta _{14}\to 0,\beta _{24}\to 0,\beta _{25}\to 0,\beta _{16}\to
0\},\{\alpha _{22},0,0\}]\text{/.}\,\{\alpha _{22}\to 0\}] }\)
\end{doublespace}

\begin{doublespace}
\noindent\(0\)
\end{doublespace}

\begin{doublespace}
\noindent\(\pmb{\text{DB25}=}
\pmb{\text{Simplify}[}
\pmb{\text{Series}[D[\text{H2},\alpha _{25}]\text{/.}\,\{\alpha _{11}\to 0,\alpha _{13}\to 0,\alpha _{24}\to 0,\beta _{11}\to
0,\beta _{21}\to 0,\beta _{12}\to 0,\beta _{22}\to 0,\beta _{13}\to 0,}\\
\pmb{\beta _{23}\to 0,\beta _{14}\to 0,\beta _{24}\to 0,\beta _{25}\to 0,\beta _{16}\to
0\},\{\alpha _{22},0,0\}]\text{/.}\,\{\alpha _{22}\to 0\}] }\)
\end{doublespace}

\begin{doublespace}
\noindent\(0\)
\end{doublespace}

\begin{doublespace}
\noindent\(\pmb{\text{DB16}=}
\pmb{\text{Simplify}[}
\pmb{\text{Series}[D[\text{H2},\alpha _{16}]\text{/.}\,\{\alpha _{11}\to 0,\alpha _{13}\to 0,\alpha _{24}\to 0,\beta _{11}\to
0,\beta _{21}\to 0,\beta _{12}\to 0,\beta _{22}\to 0,\beta _{13}\to 0,}\\
\pmb{\beta _{23}\to 0,\beta _{14}\to 0,\beta _{24}\to 0,\beta _{25}\to 0,\beta _{16}\to
0\},\{\alpha _{22},0,0\}]\text{/.}\,\{\alpha _{22}\to 0\}] }\)
\end{doublespace}

\begin{doublespace}
\noindent\(0\)
\end{doublespace}

(* Next we find the (right parts of the) reduced equations (from the full Hamiltonian)  for the remainder of the variables, calculate the reduced Hamiltonian and check that the equations of the reduced system are the same as the equations from the reduced Hamiltonian   *)

\begin{doublespace}
\noindent\(\pmb{\text{DA12}=}
\pmb{\text{Simplify}[}
\pmb{\text{Series}[D[\text{H2},\beta _{12}]\text{/.}\,\{\alpha _{11}\to 0,\alpha _{13}\to 0,\alpha _{24}\to 0,\beta _{11}\to 0,\beta
_{21}\to 0,\beta _{12}\to 0,\beta _{22}\to 0,\beta _{13}\to 0,}\\
\pmb{\beta _{23}\to 0,\beta _{14}\to 0,\beta _{24}\to 0,\beta _{25}\to 0,\beta _{16}\to
0\},\{\alpha _{22},0,0\}]\text{/.}\,\{\alpha _{22}\to 0\}] }\)
\end{doublespace}

\begin{doublespace}
\noindent\(\frac{1}{(a_2-a_4){}^3 (a_2-a_5){}^3 (a_4-a_5){}^2}2 \alpha _{12} (-a_2 (-a_5^3 (b_2+b_4)+a_4^3
(b_2+b_5)+a_2^2 (-a_5 (b_2+b_4)+a_4 (b_2+b_5))-2 a_2 (-a_5^2 (b_2+b_4)+a_4^2 (b_2+b_5)))
(a_5 (\kappa _2-\kappa _4)+a_2 (\kappa _4-\kappa _5)+a_4 (-\kappa _2+\kappa _5))+(a_5^3 (b_2+b_4)+a_2^3
(b_4-b_5)-a_4^3 (b_2+b_5)+a_2^2 (-a_5 (b_2+b_4)+a_4 (b_2+b_5))+a_2 (-a_5^2 (b_2+b_4)+a_4^2
(b_2+b_5))) (a_5^2 (\kappa _2-\kappa _4)+a_2^2 (\kappa _4-\kappa _5)+a_4^2 (-\kappa _2+\kappa
_5))-(a_5^2 (b_2+b_4)+a_2^2 (b_4-b_5)-a_4^2 (b_2+b_5)+2 a_2 (-a_5 (b_2+b_4)+a_4
(b_2+b_5))) (a_2 a_5 (-a_2+a_5) \kappa _4+a_4^2 (a_5 \kappa _2-a_2 \kappa _5)+a_4 (-a_5^2
\kappa _2+a_2^2 \kappa _5)))\)
\end{doublespace}

\begin{doublespace}
\noindent\(\pmb{\text{DA14}=}
\pmb{\text{Simplify}[}
\pmb{\text{Series}[D[\text{H2},\beta _{14}]\text{/.}\,\{\alpha _{11}\to 0,\alpha _{13}\to 0,\alpha _{24}\to 0,\beta _{11}\to 0,\beta
_{21}\to 0,\beta _{12}\to 0,\beta _{22}\to 0,\beta _{13}\to 0,}\\
\pmb{\beta _{23}\to 0,\beta _{14}\to 0,\beta _{24}\to 0,\beta _{25}\to 0,\beta _{16}\to
0\},\{\alpha _{22},0,0\}]\text{/.}\,\{\alpha _{22}\to 0\}]}\)
\end{doublespace}

\begin{doublespace}
\noindent\(\frac{1}{(a_2-a_4){}^3 (a_2-a_5){}^2 (a_4-a_5){}^3}2 \alpha _{14} (-a_4 (a_5^3 (b_2+b_4)-a_2^3
(b_4+b_5)+a_4^2 (a_5 (b_2+b_4)-a_2 (b_4+b_5))-2 a_4 (a_5^2 (b_2+b_4)-a_2^2 (b_4+b_5)))
(a_5 (-\kappa _2+\kappa _4)+a_4 (\kappa _2-\kappa _5)+a_2 (-\kappa _4+\kappa _5))-(a_5^3 (b_2+b_4)+a_4^3
(b_2-b_5)-a_2^3 (b_4+b_5)+a_4^2 (-a_5 (b_2+b_4)+a_2 (b_4+b_5))+a_4 (-a_5^2 (b_2+b_4)+a_2^2
(b_4+b_5))) (a_5^2 (-\kappa _2+\kappa _4)+a_4^2 (\kappa _2-\kappa _5)+a_2^2 (-\kappa _4+\kappa
_5))+(-a_5^2 (b_2+b_4)+a_4^2 (-b_2+b_5)+a_2^2 (b_4+b_5)+2 a_4 (a_5 (b_2+b_4)-a_2
(b_4+b_5))) (a_2 a_5 (-a_2+a_5) \kappa _4+a_4^2 (a_5 \kappa _2-a_2 \kappa _5)+a_4 (-a_5^2
\kappa _2+a_2^2 \kappa _5)))\)
\end{doublespace}
\begin{doublespace}
\noindent\(\pmb{\text{Da1}=\text{Simplify}[\text{Series}[D[\text{H2},\kappa _1]\text{/.}\{\alpha _{11}\to 0,\alpha _{13}\to
0,\alpha _{24}\to 0,\beta _{21}\to 0,\beta _{12}\to 0,\beta _{23}\to 0,\beta _{14}\to 0,}\\
\pmb{\beta _{25}\to 0,\beta _{16}\to 0\},\{\alpha _{22},0,0\}]\text{/.}\,\{\alpha _{22}\to 0\}]\text{   }}\)
\end{doublespace}

\begin{doublespace}
\noindent\(\frac{1}{(a_1-a_3){}^2 (a_1-a_6){}^2 (a_3-a_6)}2 (a_3^3 (\kappa _1-\kappa _6)+a_3^2 (a_6
(3 \kappa _1-2 \kappa _6)-a_1 \kappa _6)+a_6 (a_1 a_6 \kappa _3+a_6^2 (-\kappa _1+\kappa _3)+a_1^2 (-2 \kappa
_3+\kappa _6))+a_3 (a_6^2 (-3 \kappa _1+2 \kappa _3)-a_1^2 (\kappa _3-2 \kappa _6)+a_1 a_6 (-\kappa _3+\kappa
_6)))\)
\end{doublespace}

\begin{doublespace}
\noindent\(\pmb{\text{Da2}=}
\pmb{\text{Simplify}[}
\pmb{\text{Series}[D[\text{H2},\kappa _2]\text{/.}\,\{\alpha _{11}\to 0,\alpha _{13}\to 0,\alpha _{24}\to 0,\beta _{11}\to 0,\beta
_{21}\to 0,\beta _{12}\to 0,\beta _{22}\to 0,}\\
\pmb{\beta _{13}\to 0,\beta _{23}\to 0,\beta _{14}\to 0,\beta _{24}\to 0,\beta _{25}\to 0,\beta _{16}\to
0\},\{\alpha _{22},0,0\}]\text{/.}\,\{\alpha _{22}\to 0\}]\text{   }}\)
\end{doublespace}

\begin{doublespace}
\noindent\(\frac{1}{(a_2-a_4){}^2 (a_2-a_5){}^2 (a_4-a_5)}2 (a_4^3 (\kappa _2-\kappa _5)+a_4^2 (a_5
(3 \kappa _2-2 \kappa _5)-a_2 \kappa _5)+a_5 (a_2 a_5 \kappa _4+a_5^2 (-\kappa _2+\kappa _4)+a_2^2 (-2 \kappa
_4+\kappa _5))+a_4 (a_5^2 (-3 \kappa _2+2 \kappa _4)-a_2^2 (\kappa _4-2 \kappa _5)+a_2 a_5 (-\kappa _4+\kappa
_5)))\)
\end{doublespace}

\begin{doublespace}
\noindent\(\pmb{\text{Da3}=}
\pmb{\text{Simplify}[}
\pmb{\text{Series}[D[\text{H2},\kappa _3]\text{/.}\,\{\alpha _{11}\to 0,\alpha _{13}\to 0,\alpha _{24}\to 0,\beta _{11}\to 0,\beta
_{21}\to 0,\beta _{12}\to 0,\beta _{22}\to 0,}\\
\pmb{\beta _{13}\to 0,\beta _{23}\to 0,\beta _{14}\to 0,\beta _{24}\to 0,\beta _{25}\to 0,\beta _{16}\to
0\},\{\alpha _{22},0,0\}]\text{/.}\,\{\alpha _{22}\to 0\}]\text{  }\text{(*}\text{*)}}\)
\end{doublespace}

\begin{doublespace}
\noindent\(\frac{1}{(a_1-a_3){}^2 (a_1-a_6) (a_3-a_6){}^2}2 (a_1^3 (\kappa _3-\kappa _6)+a_1^2 (a_6
(3 \kappa _3-2 \kappa _6)-a_3 \kappa _6)+a_6 (a_3 a_6 \kappa _1+a_6^2 (\kappa _1-\kappa _3)+a_3^2 (-2 \kappa
_1+\kappa _6))+a_1 (a_6^2 (2 \kappa _1-3 \kappa _3)-a_3^2 (\kappa _1-2 \kappa _6)+a_3 a_6 (-\kappa _1+\kappa
_6)))\)
\end{doublespace}

\begin{doublespace}
\noindent\(\pmb{\text{Da4}=}
\pmb{\text{Simplify}[}
\pmb{\text{Series}[D[\text{H2},\kappa _4]\text{/.}\,\{\alpha _{11}\to 0,\alpha _{13}\to 0,\alpha _{24}\to 0,\beta _{11}\to 0,\beta
_{21}\to 0,\beta _{12}\to 0,\beta _{22}\to 0,}\\
\pmb{\beta _{13}\to 0,\beta _{23}\to 0,\beta _{14}\to 0,\beta _{24}\to 0,\beta _{25}\to 0,\beta _{16}\to
0\},}
\pmb{\{\alpha _{22},0,0\}]\text{/.}\,\{\alpha _{22}\to 0\}]}\)
\end{doublespace}

\begin{doublespace}
\noindent\(\frac{1}{(a_2-a_4){}^2 (a_2-a_5) (a_4-a_5){}^2}2 (a_2^3 (\kappa _4-\kappa _5)+a_2^2 (a_5
(3 \kappa _4-2 \kappa _5)-a_4 \kappa _5)+a_5 (a_4 a_5 \kappa _2+a_5^2 (\kappa _2-\kappa _4)+a_4^2 (-2 \kappa
_2+\kappa _5))+a_2 (a_5^2 (2 \kappa _2-3 \kappa _4)-a_4^2 (\kappa _2-2 \kappa _5)+a_4 a_5 (-\kappa _2+\kappa
_5)))\)
\end{doublespace}

\begin{doublespace}
\noindent\(\pmb{\text{Da5}=}
\pmb{\text{Simplify}[}
\pmb{\text{Series}[D[\text{H2},\kappa _5]\text{/.}\,\{\alpha _{11}\to 0,\alpha _{13}\to 0,\alpha _{24}\to 0,\beta _{11}\to 0,\beta
_{21}\to 0,\beta _{12}\to 0,\beta _{22}\to 0,}\\
\pmb{\beta _{13}\to 0,\beta _{23}\to 0,\beta _{14}\to 0,\beta _{24}\to 0,\beta _{25}\to 0,\beta _{16}\to
0\},\{\alpha _{22},0,0\}]\text{/.}\,\{\alpha _{22}\to 0\}]}\)
\end{doublespace}

\begin{doublespace}
\noindent\(-\frac{1}{(a_2-a_4) (a_2-a_5){}^2 (a_4-a_5){}^2}2 (a_2^2 (a_5 \kappa _4+a_4 (2 \kappa
_4-3 \kappa _5))-a_4 (a_4 a_5 \kappa _2+a_5^2 (-2 \kappa _2+\kappa _4)+a_4^2 (\kappa _2-\kappa _5))+a_2^3
(\kappa _4-\kappa _5)+a_2 (a_5^2 (\kappa _2-2 \kappa _4)+a_4 a_5 (\kappa _2-\kappa _4)+a_4^2 (-2 \kappa
_2+3 \kappa _5))) \)
\end{doublespace}

\begin{doublespace}
\noindent\(\pmb{\text{Da6}=}
\pmb{\text{Simplify}[}
\pmb{\text{Series}[D[\text{H2},\kappa _6]\text{/.}\,\{\alpha _{11}\to 0,\alpha _{13}\to 0,\alpha _{24}\to 0,\beta _{11}\to 0,\beta
_{21}\to 0,\beta _{12}\to 0,\beta _{22}\to 0,}\\
\pmb{\beta _{13}\to 0,\beta _{23}\to 0,\beta _{14}\to 0,\beta _{24}\to 0,\beta _{25}\to 0,\beta _{16}\to
0\},\{\alpha _{22},0,0\}]\text{/.}\,\{\alpha _{22}\to 0\}]}\)
\end{doublespace}

\begin{doublespace}
\noindent\(-\frac{1}{(a_1-a_3) (a_1-a_6){}^2 (a_3-a_6){}^2}2 (a_1^2 (a_6 \kappa _3+a_3 (2 \kappa
_3-3 \kappa _6))-a_3 (a_3 a_6 \kappa _1+a_6^2 (-2 \kappa _1+\kappa _3)+a_3^2 (\kappa _1-\kappa _6))+a_1^3
(\kappa _3-\kappa _6)+a_1 (a_6^2 (\kappa _1-2 \kappa _3)+a_3 a_6 (\kappa _1-\kappa _3)+a_3^2 (-2 \kappa
_1+3 \kappa _6))) \)
\end{doublespace}

\begin{doublespace}
\noindent\(\pmb{\text{Simplify}[\text{Da1}+\text{Da3}+\text{Da6}]}\)
\end{doublespace}

\begin{doublespace}
\noindent\(\frac{2 (a_6 (-\kappa _1+\kappa _3)+a_3 (\kappa _1-\kappa _6)+a_1 (-\kappa _3+\kappa _6))}{(a_1-a_3)
(a_1-a_6) (a_3-a_6)}\)
\end{doublespace}

\begin{doublespace}
\noindent\(\pmb{\text{(*} \text{The dependence of}\ a_s\ \text{on}\ b_s\ \text{has been switched on at calculating the}\ \text{D$\kappa $}_s.\ \text{It does not affect the} }\\
\pmb{\text{result}:\text{  }\text{*)}}\)
\end{doublespace}

\begin{doublespace}
\noindent\(\pmb{\text{D$\kappa $1}=}
\pmb{\text{Simplify}[}
\pmb{\text{Series}[-D[\text{H2}\text{/.}\,\{b_1\to f[a_1]\},a_1]\text{/.}\,\{\alpha _{11}\to 0,\alpha _{13}\to
0,\alpha _{24}\to 0,\beta _{11}\to 0,\beta _{21}\to 0,}\\
\pmb{\beta _{12}\to 0,\beta _{22}\to 0,\beta _{13}\to 0,\beta _{23}\to 0,\beta _{14}\to 0,\beta _{24}\to 0,\beta _{25}\to 0,\beta _{16}\to 0\},\{\alpha _{22},0,0\}]\text{/.}\,\{\alpha
_{22}\to 0\}]}\)
\end{doublespace}

\begin{doublespace}
\noindent\(\frac{1}{(a_1-a_3){}^3 (a_1-a_6){}^3 (a_3-a_6){}^2}2 (-a_6^5 (\kappa _1-\kappa _3){}^2+a_3
a_6^4 (-3 \kappa _1^2+5 \kappa _1 \kappa _3-2 \kappa _3^2)+a_3^2 a_6^3 (\kappa _1-\kappa _3) (4 \kappa _1-3 \kappa _6)+a_3^3
a_6^2 (4 \kappa _1-3 \kappa _3) (\kappa _1-\kappa _6)-a_3^5 (\kappa _1-\kappa _6){}^2+a_1^4 (a_3 (\kappa
_3-2 \kappa _6)+a_6 (2 \kappa _3-\kappa _6)) (\kappa _3-\kappa _6)-3 a_1^2 (a_3-a_6) (a_6 (-\kappa
_1+\kappa _3)+a_3 (\kappa _1-\kappa _6)) (a_6 \kappa _3+a_3 \kappa _6)+a_3^4 a_6 (-3 \kappa _1^2+5 \kappa
_1 \kappa _6-2 \kappa _6{}^2)+a_1 (a_6^4 (2 \kappa _1^2-3 \kappa _1 \kappa _3+\kappa _3^2)+a_3^3 a_6 (4 \kappa _1+3 \kappa
_3-5 \kappa _6) (\kappa _1-\kappa _6)+a_3 a_6^3 (\kappa _1-\kappa _3) (4 \kappa _1-5 \kappa _3+3 \kappa _6)+a_3^4
(2 \kappa _1^2-3 \kappa _1 \kappa _6+\kappa _6{}^2)+3 a_3^2 a_6^2 (-4 \kappa _1^2-2 \kappa _3 \kappa _6+3 \kappa _1 (\kappa
_3+\kappa _6)))-a_1^3 (a_3 a_6 (\kappa _3^2-6 \kappa _3 \kappa _6+\kappa _6{}^2+2 \kappa _1 (\kappa _3+\kappa _6))+a_6^2
(\kappa _3 (5 \kappa _3-3 \kappa _6)+\kappa _1 (-4 \kappa _3+2 \kappa _6))+a_3^2 (2 \kappa _1 (\kappa
_3-2 \kappa _6)+\kappa _6 (-3 \kappa _3+5 \kappa _6))))\)
\end{doublespace}

\begin{doublespace}
\noindent\(\pmb{\text{D$\kappa $2}=}
\pmb{\text{Simplify}[}
\pmb{\text{Series}[-D[\text{H2}\text{/.}\,\{b_2\to f[a_2]\},a_2]\text{/.}\,\{\alpha _{11}\to 0,\alpha _{13}\to
0,\alpha _{24}\to 0,\beta _{11}\to 0,\beta _{21}\to 0,}\\
\pmb{\beta _{12}\to 0,\beta _{22}\to 0,\beta _{13}\to 0,\beta _{23}\to 0,\beta _{14}\to 0,\beta _{24}\to 0,\beta _{25}\to 0,\beta _{16}\to 0\},\{\alpha _{22},0,0\}]\text{/.}\,\{\alpha
_{22}\to 0\}]}\)
\end{doublespace}

\begin{doublespace}
\noindent\(\frac{1}{(a_2-a_4){}^3 (a_2-a_5){}^3 (a_4-a_5){}^2}2 (-a_5^5 (\kappa _2-\kappa _4){}^2+a_4
a_5^4 (-3 \kappa _2^2+5 \kappa _2 \kappa _4-2 \kappa _4^2)+a_4^2 a_5^3 (\kappa _2-\kappa _4) (4 \kappa _2-3 \kappa _5)+a_4^3
a_5^2 (4 \kappa _2-3 \kappa _4) (\kappa _2-\kappa _5)-a_4^5 (\kappa _2-\kappa _5){}^2+a_2^4 (a_4 (\kappa
_4-2 \kappa _5)+a_5 (2 \kappa _4-\kappa _5)) (\kappa _4-\kappa _5)-3 a_2^2 (a_4-a_5) (a_5 (-\kappa
_2+\kappa _4)+a_4 (\kappa _2-\kappa _5)) (a_5 \kappa _4+a_4 \kappa _5)+a_4^4 a_5 (-3 \kappa _2^2+5 \kappa
_2 \kappa _5-2 \kappa _5{}^2)+a_2 (a_5^4 (2 \kappa _2^2-3 \kappa _2 \kappa _4+\kappa _4^2)+a_4^3 a_5 (4 \kappa _2+3 \kappa
_4-5 \kappa _5) (\kappa _2-\kappa _5)+a_4 a_5^3 (\kappa _2-\kappa _4) (4 \kappa _2-5 \kappa _4+3 \kappa _5)+a_4^4
(2 \kappa _2^2-3 \kappa _2 \kappa _5+\kappa _5{}^2)+3 a_4^2 a_5^2 (-4 \kappa _2^2-2 \kappa _4 \kappa _5+3 \kappa _2 (\kappa
_4+\kappa _5)))-a_2^3 (a_4 a_5 (\kappa _4^2-6 \kappa _4 \kappa _5+\kappa _5{}^2+2 \kappa _2 (\kappa _4+\kappa _5))+a_5^2
(\kappa _4 (5 \kappa _4-3 \kappa _5)+\kappa _2 (-4 \kappa _4+2 \kappa _5))+a_4^2 (2 \kappa _2 (\kappa
_4-2 \kappa _5)+\kappa _5 (-3 \kappa _4+5 \kappa _5))))\)
\end{doublespace}

\begin{doublespace}
\noindent\(\pmb{\text{D$\kappa $3}=}
\pmb{\text{Simplify}[}
\pmb{\text{Series}[-D[\text{H2}\text{/.}\,\{b_3\to f[a_3]\},a_3]\text{/.}\,\{\alpha _{11}\to 0,\alpha _{13}\to
0,\alpha _{24}\to 0,\beta _{11}\to 0,\beta _{21}\to 0,}\\
\pmb{\beta _{12}\to 0,\beta _{22}\to 0,\beta _{13}\to 0,\beta _{23}\to 0,\beta _{14}\to 0,\beta _{24}\to 0,\beta _{25}\to 0,\beta _{16}\to 0\},\{\alpha _{22},0,0\}]\text{/.}\,\{\alpha
_{22}\to 0\}]}\)
\end{doublespace}

\begin{doublespace}
\noindent\(\frac{1}{(a_1-a_3){}^3 (a_1-a_6){}^2 (a_3-a_6){}^3}2 (a_6^2 (\kappa _1-\kappa _3)-2 a_3
(a_6 (\kappa _1-\kappa _3)+a_1 (\kappa _3-\kappa _6))+a_3^2 (\kappa _1-\kappa _6)+a_1^2 (\kappa
_3-\kappa _6)) (a_1^3 (\kappa _3-\kappa _6)+a_1^2 (a_6 (3 \kappa _3-2 \kappa _6)-a_3 \kappa _6)+a_6
(a_3 a_6 \kappa _1+a_6^2 (\kappa _1-\kappa _3)+a_3^2 (-2 \kappa _1+\kappa _6))+a_1 (a_6^2 (2 \kappa _1-3
\kappa _3)-a_3^2 (\kappa _1-2 \kappa _6)+a_3 a_6 (-\kappa _1+\kappa _6)))\)
\end{doublespace}

\begin{doublespace}
\noindent\(\pmb{\text{(* The obtained expression for}\ \text{D$\kappa $3}\ \text{is simpler than for other $\kappa$'s. This enabled us to simplify } }\\
\pmb{\text{the others and finally to conjecture the relations between the derivatives of}\ \kappa\  \text{and}\ a\ \text{*)}}\)
\end{doublespace}

\begin{doublespace}
\noindent\(\pmb{\text{D$\kappa $4}=}
\pmb{\text{Simplify}[}
\pmb{\text{Series}[-D[\text{H2}\text{/.}\,\{b_4\to f[a_4]\},a_4]\text{/.}\,\{\alpha _{11}\to 0,\alpha _{13}\to
0,\alpha _{24}\to 0,\beta _{11}\to 0,\beta _{21}\to 0,}\\
\pmb{\beta _{12}\to 0,\beta _{22}\to 0,\beta _{13}\to 0,\beta _{23}\to 0,\beta _{14}\to 0,\beta _{24}\to 0,\beta _{25}\to 0,\beta _{16}\to 0\},\{\alpha _{22},0,0\}]\text{/.}\,\{\alpha
_{22}\to 0\}]}\)
\end{doublespace}

\begin{doublespace}
\noindent\(\frac{1}{(a_2-a_4){}^3 (a_2-a_5){}^2 (a_4-a_5){}^3}2 (a_2^5 (\kappa _4-\kappa _5){}^2-a_2^4
(\kappa _4-\kappa _5) (a_4 (2 \kappa _4-\kappa _5)+a_5 (-3 \kappa _4+2 \kappa _5))+a_2^3 (\kappa
_4-\kappa _5) (a_5^2 (3 \kappa _2-4 \kappa _4)+3 a_4^2 \kappa _5+a_4 a_5 (-3 \kappa _2-4 \kappa _4+5 \kappa _5))+a_2
(a_5^4 (2 \kappa _2^2-5 \kappa _2 \kappa _4+3 \kappa _4^2)-a_4 a_5^3 (\kappa _2-\kappa _4) (5 \kappa _2-4 \kappa
_4-3 \kappa _5)+3 a_4^2 a_5^2 (\kappa _2^2-2 \kappa _2 \kappa _4+\kappa _4 \kappa _5)-a_4^4 (\kappa _2^2-3 \kappa _2 \kappa
_5+2 \kappa _5{}^2)+a_4^3 a_5 (\kappa _2^2+2 \kappa _2 (\kappa _4-3 \kappa _5)+\kappa _5 (2 \kappa _4+\kappa _5)))+a_5
(a_5^4 (\kappa _2-\kappa _4){}^2+3 a_4^2 a_5^2 \kappa _2 (-\kappa _2+\kappa _4)-a_4 a_5^3 (\kappa _2^2-3 \kappa _2
\kappa _4+2 \kappa _4^2)-a_4^4 (2 \kappa _2^2-3 \kappa _2 \kappa _5+\kappa _5{}^2)+a_4^3 a_5 (5 \kappa _2^2+2 \kappa _4 \kappa
_5-\kappa _2 (4 \kappa _4+3 \kappa _5)))+a_2^2 (a_5^3 (\kappa _2-\kappa _4) (4 \kappa _4-3 \kappa _5)+3
a_4^2 a_5 (\kappa _2 \kappa _4+\kappa _5 (-2 \kappa _4+\kappa _5))+3 a_4 a_5^2 (\kappa _4 (4 \kappa _4-3 \kappa _5)+\kappa
_2 (-3 \kappa _4+2 \kappa _5))+a_4^3 (\kappa _2 (2 \kappa _4-3 \kappa _5)+\kappa _5 (-4 \kappa _4+5 \kappa
_5))))\)
\end{doublespace}

\begin{doublespace}
\noindent\(\pmb{\text{D$\kappa $5}=}
\pmb{\text{Simplify}[}
\pmb{\text{Series}[-D[\text{H2}\text{/.}\,\{b_5\to f[a_5]\},a_5]\text{/.}\,\{\alpha _{11}\to 0,\alpha _{13}\to
0,\alpha _{24}\to 0,\beta _{11}\to 0,\beta _{21}\to 0,}\\
\pmb{\beta _{12}\to 0,\beta _{22}\to 0,\beta _{13}\to 0,\beta _{23}\to 0,\beta _{14}\to 0,\beta _{24}\to 0,\beta _{25}\to 0,\beta _{16}\to 0\},\{\alpha _{22},0,0\}]\text{/.}\,\{\alpha
_{22}\to 0\}]}\)
\end{doublespace}

\begin{doublespace}
\noindent\(-\frac{1}{(a_2-a_4){}^2 (a_2-a_5){}^3 (a_4-a_5){}^3}2 (-a_2^3 (3 a_5^2 \kappa _4+a_4^2 (3
\kappa _2-4 \kappa _5)+a_4 a_5 (-3 \kappa _2+5 \kappa _4-4 \kappa _5)) (\kappa _4-\kappa _5)+a_2^4 (a_4 (2
\kappa _4-3 \kappa _5)-a_5 (\kappa _4-2 \kappa _5)) (\kappa _4-\kappa _5)+a_2^5 (\kappa _4-\kappa _5){}^2+a_2
(-a_5^4 (\kappa _2^2-3 \kappa _2 \kappa _4+2 \kappa _4^2)-a_4^3 a_5 (5 \kappa _2-3 \kappa _4-4 \kappa _5) (\kappa
_2-\kappa _5)+3 a_4^2 a_5^2 (\kappa _2^2-2 \kappa _2 \kappa _5+\kappa _4 \kappa _5)+a_4^4 (2 \kappa _2^2-5 \kappa _2 \kappa
_5+3 \kappa _5{}^2)+a_4 a_5^3 (\kappa _2^2+\kappa _2 (-6 \kappa _4+2 \kappa _5)+\kappa _4 (\kappa _4+2 \kappa _5)))+a_2^2
(a_4^3 (3 \kappa _4-4 \kappa _5) (-\kappa _2+\kappa _5)+3 a_4 a_5^2 (\kappa _4^2+\kappa _2 \kappa _5-2 \kappa _4
\kappa _5)+a_5^3 (\kappa _4 (5 \kappa _4-4 \kappa _5)+\kappa _2 (-3 \kappa _4+2 \kappa _5))+3 a_4^2 a_5 (\kappa
_2 (2 \kappa _4-3 \kappa _5)+\kappa _5 (-3 \kappa _4+4 \kappa _5)))+a_4 (-a_5^4 (2 \kappa _2^2-3 \kappa
_2 \kappa _4+\kappa _4^2)+a_4^4 (\kappa _2-\kappa _5){}^2+3 a_4^2 a_5^2 \kappa _2 (-\kappa _2+\kappa _5)-a_4^3 a_5 (\kappa
_2^2-3 \kappa _2 \kappa _5+2 \kappa _5{}^2)+a_4 a_5^3 (5 \kappa _2^2+2 \kappa _4 \kappa _5-\kappa _2 (3 \kappa _4+4 \kappa _5))))\)
\end{doublespace}

\begin{doublespace}
\noindent\(\pmb{\text{D$\kappa $6}=}
\pmb{\text{Simplify}[}
\pmb{\text{Series}[-D[\text{H2}\text{/.}\,\{b_6\to f[a_6]\},a_6]\text{/.}\,\{\alpha _{11}\to 0,\alpha _{13}\to
0,\alpha _{24}\to 0,\beta _{11}\to 0,\beta _{21}\to 0,}\\
\pmb{\beta _{12}\to 0,\beta _{22}\to 0,\beta _{13}\to 0,\beta _{23}\to 0,\beta _{14}\to 0,\beta _{24}\to 0,\beta _{25}\to 0,\beta _{16}\to 0\},\{\alpha _{22},0,0\}]\text{/.}\,\{\alpha
_{22}\to 0\}]}\)
\end{doublespace}

\begin{doublespace}
\noindent\(-\frac{1}{(a_1-a_3){}^2 (a_1-a_6){}^3 (a_3-a_6){}^3}2 (-a_1^3 (3 a_6^2 \kappa _3+a_3^2 (3
\kappa _1-4 \kappa _6)+a_3 a_6 (-3 \kappa _1+5 \kappa _3-4 \kappa _6)) (\kappa _3-\kappa _6)+a_1^4 (a_3 (2
\kappa _3-3 \kappa _6)-a_6 (\kappa _3-2 \kappa _6)) (\kappa _3-\kappa _6)+a_1^5 (\kappa _3-\kappa _6){}^2+a_1
(-a_6^4 (\kappa _1^2-3 \kappa _1 \kappa _3+2 \kappa _3^2)-a_3^3 a_6 (5 \kappa _1-3 \kappa _3-4 \kappa _6) (\kappa
_1-\kappa _6)+3 a_3^2 a_6^2 (\kappa _1^2-2 \kappa _1 \kappa _6+\kappa _3 \kappa _6)+a_3^4 (2 \kappa _1^2-5 \kappa _1 \kappa
_6+3 \kappa _6{}^2)+a_3 a_6^3 (\kappa _1^2+\kappa _1 (-6 \kappa _3+2 \kappa _6)+\kappa _3 (\kappa _3+2 \kappa _6)))+a_1^2
(a_3^3 (3 \kappa _3-4 \kappa _6) (-\kappa _1+\kappa _6)+3 a_3 a_6^2 (\kappa _3^2+\kappa _1 \kappa _6-2 \kappa _3
\kappa _6)+a_6^3 (\kappa _3 (5 \kappa _3-4 \kappa _6)+\kappa _1 (-3 \kappa _3+2 \kappa _6))+3 a_3^2 a_6 (\kappa
_1 (2 \kappa _3-3 \kappa _6)+\kappa _6 (-3 \kappa _3+4 \kappa _6)))+a_3 (-a_6^4 (2 \kappa _1^2-3 \kappa
_1 \kappa _3+\kappa _3^2)+a_3^4 (\kappa _1-\kappa _6){}^2+3 a_3^2 a_6^2 \kappa _1 (-\kappa _1+\kappa _6)-a_3^3 a_6 (\kappa
_1^2-3 \kappa _1 \kappa _6+2 \kappa _6{}^2)+a_3 a_6^3 (5 \kappa _1^2+2 \kappa _3 \kappa _6-\kappa _1 (3 \kappa _3+4 \kappa _6))))\)
\end{doublespace}

\begin{doublespace}
\noindent\(\pmb{\text{(* The reduced Hamiltonian}: \text{*)}}\)
\end{doublespace}

\begin{doublespace}
\noindent\(\pmb{\text{H2r}=}
\pmb{\text{Simplify}[}
\pmb{\text{Cancel}[\text{H2}\text{/.}\,\{\alpha _{11}\to 0,\alpha _{13}\to 0,\alpha _{24}\to 0,\beta _{11}\to 0,\beta _{21}\to 0,\beta
_{12}\to 0,\beta _{22}\to 0,\beta _{13}\to 0,\beta _{23}\to 0,}\\
\pmb{\beta _{14}\to 0,\beta _{24}\to 0,\beta _{25}\to 0,\beta _{16}\to 0\}]]\text{/.}\,\{\alpha _{22}\to 0\}}\)
\end{doublespace}

\begin{doublespace}
\noindent\(\frac{(a_5^2 (-\kappa _2+\kappa _4)+a_4^2 (\kappa _2-\kappa _5)+a_2^2 (-\kappa _4+\kappa _5)){}^2}{(a_2-a_4){}^2
(a_2-a_5){}^2 (a_4-a_5){}^2}+\frac{2 (a_5 (-\kappa _2+\kappa _4)+a_4 (\kappa _2-\kappa _5)+a_2 (-\kappa
_4+\kappa _5)) (a_2 a_5 (-a_2+a_5) \kappa _4+a_4^2 (a_5 \kappa _2-a_2 \kappa _5)+a_4 (-a_5^2 \kappa _2+a_2^2
\kappa _5))}{(a_2-a_4){}^2 (a_2-a_5){}^2 (a_4-a_5){}^2}+\frac{(a_6^2 (-\kappa _1+\kappa _3)+a_3^2
(\kappa _1-\kappa _6)+a_1^2 (-\kappa _3+\kappa _6)){}^2}{(a_1-a_3){}^2 (a_1-a_6){}^2 (a_3-a_6){}^2}+\frac{2
(a_6 (-\kappa _1+\kappa _3)+a_3 (\kappa _1-\kappa _6)+a_1 (-\kappa _3+\kappa _6)) (a_1 a_6 (-a_1+a_6)
\kappa _3+a_3^2 (a_6 \kappa _1-a_1 \kappa _6)+a_3 (-a_6^2 \kappa _1+a_1^2 \kappa _6))}{(a_1-a_3){}^2 (a_1-a_6){}^2
(a_3-a_6){}^2}\)
\end{doublespace}

\begin{doublespace}
\noindent\(\pmb{\text{(* For example, check the coincidence Da6}==\text{Da6r}\ \text{where Da6r is obtained from the reduced Hamiltonian}: \text{*)}}\)
\end{doublespace}

\begin{doublespace}
\noindent\(\pmb{\text{Da6r}=\text{Simplify}[D[\text{H2r},\kappa _6]]}\)
\end{doublespace}

\begin{doublespace}
\noindent\(\pmb{-\frac{1}{(a_1-a_3) (a_1-a_6){}^2 (a_3-a_6){}^2}2 }\\
\pmb{(a_1^2 (a_6 \kappa _3+a_3 (2 \kappa _3-3 \kappa _6))-a_3 (a_3 a_6 \kappa  _1+a_6^2 (-2 \kappa _1+\kappa
_3)+a_3^2 (\kappa _1-\kappa _6))+a_1^3 (\kappa _3-\kappa _6)+}\\
\pmb{a_1 (a_6^2 (\kappa _1-2 \kappa _3)+a_3 a_6 (\kappa _1-\kappa _3)+a_3^2 (-2 \kappa _1+3 \kappa _6)))}\)
\end{doublespace}

\begin{doublespace}
\noindent\(\pmb{\text{Da6}==\text{Da6r}}\)
\end{doublespace}

\begin{doublespace}
\noindent\(\text{True}\)
\end{doublespace}

\begin{doublespace}
\noindent\(\pmb{\text{(* Next we find out some equations on }\
\alpha \text{ under additional requirements on $\kappa$'s.}}
\\
\pmb{ \text{We use the notation Aij for $\a_{ij}$ *)}}\)
\end{doublespace}

\begin{doublespace}
\noindent\(\pmb{\text{DA21}=}
\pmb{\text{Simplify}[\text{Series}[D[\text{H2},\beta _{21}]\text{/.}\,\{\alpha _{11}\to 0,\alpha _{13}\to 0,\alpha _{24}\to
0,\beta _{21}\to 0,\beta _{12}\to 0,\beta _{23}\to 0,}\\
\pmb{\beta _{14}\to 0,\beta _{25}\to 0,\beta _{16}\to 0\},\{\alpha _{22},0,0\}]\text{/.}\,\{\alpha _{22}\to 0,\beta _{22}\to 0,\beta _{24}\to 0\}]\text{/.}\,\{\kappa _3\to \kappa _1,\kappa _6\to \kappa _1\}}\)
\end{doublespace}

\begin{doublespace}
\noindent\(\frac{1}{(a_1-a_3){}^3 (a_1-a_6){}^3 (a_3-a_6){}^2}2 \alpha _{21} ((a_6^2 (b_1+b_3)+a_1^2
(b_3-b_6)-a_3^2 (b_1+b_6)+2 a_1 (-a_6 (b_1+b_3)+a_3 (b_1+b_6))) (a_3^2 (a_1
\kappa _1-a_6 \kappa _1)+a_3 (-a_1^2 \kappa _1+a_6^2 \kappa _1)+a_1 (a_1-a_6) a_6 \kappa _3)+a_1 (a_6^3 (b_1+b_3)-a_3^3
(b_1+b_6)+a_1^2 (a_6 (b_1+b_3)-a_3 (b_1+b_6))+2 a_1 (-a_6^2 (b_1+b_3)+a_3^2 (b_1+b_6)))
(a_6 (\kappa _1-\kappa _3)+a_1 (-\kappa _1+\kappa _3))+(a_6^3 (b_1+b_3)+a_1^3 (b_3-b_6)-a_3^3
(b_1+b_6)+a_1^2 (-a_6 (b_1+b_3)+a_3 (b_1+b_6))+a_1 (-a_6^2 (b_1+b_3)+a_3^2 (b_1+b_6)))
(a_6^2 (\kappa _1-\kappa _3)+a_1^2 (-\kappa _1+\kappa _3)))\)
\end{doublespace}

\begin{doublespace}
\noindent\(\pmb{\text{DA21}=\text{Simplify}[\text{DA21}\text{/.}\,\{\kappa _1\to K,\kappa _3\to K,\kappa _6\to K\}]\text{
}\text{(*  we will regard to K as to a constant   *)}}\)
\end{doublespace}

\begin{doublespace}
\noindent\(-\frac{2 K (a_6^2 (b_1+b_3)+a_1^2 (b_3-b_6)-a_3^2 (b_1+b_6)+2 a_1 (-a_6 (b_1+b_3)+a_3
(b_1+b_6))) \alpha _{21}}{(a_1-a_3){}^2 (a_1-a_6){}^2 (a_3-a_6)}\)
\end{doublespace}

\begin{doublespace}
\noindent\(\pmb{\text{(* after a hand made simplification we obtain}:\text{  }\text{*)}}\)
\end{doublespace}

\begin{doublespace}
\noindent\(\pmb{\text{DA21}\text{:=}\frac{2 K ((a_1-a_3){}^2 (b_1+b_6)-(a_1-a_6){}^2 (b_1+b_3))
\alpha _{21}}{(a_1-a_3){}^2 (a_1-a_6){}^2 (a_3-a_6)}}\)
\end{doublespace}

\begin{doublespace}
\noindent\(\pmb{\text{(*}\text{  }\text{Finally} \text{DA21}\text{:=}\frac{2\text{K$\alpha $}_{21}(b_1+b_6) }{ (a_1-a_6){}^2(a_3-a_6)}-\frac{2\text{K$\alpha
$}_{21}(b_1+b_3) }{(a_1-a_3){}^2(a_3-a_6)}\text{   }\text{*)}}\)
\end{doublespace}

\begin{doublespace}
\noindent\(\pmb{\text{DA23}=\text{Simplify}[\text{Series}[D[\text{H2},\beta _{23}]\text{/.}\,\{\alpha _{11}\to 0,\alpha _{13}\to
0,\alpha _{24}\to 0,\beta _{21}\to 0,\beta _{12}\to 0,\beta _{23}\to 0,}\\
\pmb{\beta _{14}\to 0,\beta _{25}\to 0,\beta _{16}\to 0\},\{\alpha _{22},0,0\}]\text{/.}\,\{\alpha _{22}\to 0,\beta _{22}\to 0,\beta _{24}\to 0\}]}\)
\end{doublespace}

\begin{doublespace}
\noindent\(-\frac{1}{(a_1-a_3){}^3 (a_1-a_6){}^2 (a_3-a_6){}^3}2 \alpha _{23} (a_3 (a_6^3 (b_1+b_3)-a_1^3
(b_3+b_6)+a_3^2 (a_6 (b_1+b_3)-a_1 (b_3+b_6))-2 a_3 (a_6^2 (b_1+b_3)-a_1^2 (b_3+b_6)))
(a_6 (-\kappa _1+\kappa _3)+a_3 (\kappa _1-\kappa _6)+a_1 (-\kappa _3+\kappa _6))+(a_6^3 (b_1+b_3)+a_3^3
(b_1-b_6)-a_1^3 (b_3+b_6)+a_3^2 (-a_6 (b_1+b_3)+a_1 (b_3+b_6))+a_3 (-a_6^2 (b_1+b_3)+a_1^2
(b_3+b_6))) (a_6^2 (-\kappa _1+\kappa _3)+a_3^2 (\kappa _1-\kappa _6)+a_1^2 (-\kappa _3+\kappa
_6))+(a_6^2 (b_1+b_3)+a_3^2 (b_1-b_6)-a_1^2 (b_3+b_6)-2 a_3 (a_6 (b_1+b_3)-a_1
(b_3+b_6))) (a_1 a_6 (-a_1+a_6) \kappa _3+a_3^2 (a_6 \kappa _1-a_1 \kappa _6)+a_3 (-a_6^2
\kappa _1+a_1^2 \kappa _6)))\)
\end{doublespace}

\begin{doublespace}
\noindent\(\pmb{\text{(*}\text{  }\text{Finally}\text{   }\text{DA23}=\frac{2\text{K$\alpha $}_{23} ((a_1-a_3){}^2(b_3+b_6)-(a_3-a_6){}^2(b_1+b_3)))
}{(a_1-a_3){}^2 (a_1-a_6) (a_3-a_6){}^2}=\frac{2\text{K$\alpha $}_{23}(b_3+b_6) }{(a_1-a_6)(a_3-a_6){}^2}-\frac{2\text{K$\alpha
$}_{23}(b_1+b_3) }{(a_1-a_3){}^2(a_1-a_6) }\text{   }\text{*)}}\)
\end{doublespace}
}
\bibliographystyle{amsalpha}


\end{document}